\documentclass[preprint 12pt amsrefs]{elsarticle}
\usepackage{float}
\usepackage{array}     
\usepackage{colortbl}     
  \usepackage{arydshln}

\usepackage{footmisc}
\biboptions{sort&compress} 

\usepackage{hyperref}
\hypersetup{colorlinks=false  }

\usepackage{graphicx}
\graphicspath{{../images/}{../PODS08/}{../IDEAS08/}{../TagClouds/latex}}
\usepackage{calc}
\usepackage{}
\usepackage{amssymb}
\usepackage{amsmath}
\usepackage[english]{babel}
\usepackage{setspace}
\usepackage[vlined,algonl,oldcommands,noend,boxed]{algorithm2e} 

\usepackage{url}
\usepackage{multirow}
\usepackage{rotating}
\usepackage[T1]{fontenc}
\usepackage{framed}
\usepackage{times,verbatim} 
\usepackage{listings}
\usepackage{amsthm}
\usepackage{fancyhdr}
\usepackage{algorithmic}
\usepackage{fancyvrb} 
\usepackage{epsfig}
\usepackage{subfig}

\usepackage[table]{xcolor}
\definecolor{grayish}{rgb}{0.9,0.9,0.9}
\definecolor{blueish}{rgb}{0.6,0.6,1.0}
\definecolor{reddish}{rgb}{1.0,0.6,0.6}
\definecolor{dkgrey}{rgb}{0.6,0.6,0.6}

\newcommand{\varfont}[1]{\mathtt{#1}}
\newcommand{\constfont}[1] {\mbox{\textbf{#1}}}

\newtheorem{definition}{Definition}[section]
\newtheorem{theorem}{Theorem}[section]

\newtheorem{proposition}{Proposition}[section]

\newcommand{\diamondcube}{diamond}
\newcommand{\Diamondcube}{Diamond}

\newcommand{\imd}{IMD}
\newcommand{\emd}{EMD} 
\newcommand{\algosize}{\footnotesize}

\newcommand{\dblpCubeOne}{\textsf{D1}}
\newcommand{\dblpCubeTwo}{\textsf{D2}}
\newcommand{\dblpCubeThree}{\textsf{D3}}

\newcommand{\twCubeOne}{\textsf{TW1}}
\newcommand{\twCubeTwo}{\textsf{TW2}}
\newcommand{\twCubeThree}{\textsf{TW3}}

\newcommand{\nfCubeOne}{\textsf{NF1}}
\newcommand{\nfCubeTwo}{\textsf{NF2}}
\newcommand{\nfCubeThree}{\textsf{NF3}}

\newcommand{\cCubeOne}{\textsf{C1}}
\newcommand{\cCubeTwo}{\textsf{C2}}

\newcommand{\bCubeOne}{\textsf{B1}}
\newcommand{\bCubeTwo}{\textsf{B2}}
\newcommand{\bCubeThree}{\textsf{B3}}
\newcommand{\bCubeFour}{\textsf{B4}}
\newcommand{\bCubeFive}{\textsf{B5}}
\newcommand{\bCubeEight}{\textsf{B6}}

\newcommand{\wCubeFour}{\textsf{W1}}
\newcommand{\wCubeTwo}{\textsf{W2}}

\newcommand{\ssbCubeOne}{\textsf{SSB1}}

\newcommand{\uCubeOne}{\textsf{U1}}
\newcommand{\uCubeTwo}{\textsf{U2}}
\newcommand{\uCubeThree}{\textsf{U3}}

\newcommand{\sCubeOne}{\textsf{S1}}
\newcommand{\sCubeTwo}{\textsf{S2}}
\newcommand{\sCubeThree}{\textsf{S3}}

\newcommand{\ssCubeOne}{\textsf{SS1}}
\newcommand{\ssCubeTwo}{\textsf{SS2}}
\newcommand{\ssCubeThree}{\textsf{SS3}}

\usepackage{todonotes}

\newcommand{\hazel}[1]{#1}
\newcommand{\owen}[1]{#1}

\journal{Data and Knowledge Engineering}

\author[unb]{Hazel Webb}
\address [unb]{University of New Brunswick Saint John}
\ead{hazel.webb@unb.ca}
\ead[url]{http://hazel-webb.com}

\author[uq]{ Daniel Lemire}
\address[uq]{TELUQ, Universit\'{e} du Qu\'{e}bec}
\ead{lemire@gmail.com}

\author[unb]{Owen Kaser}
\ead{owen@computer.org}

\begin{document}

\begin{frontmatter}
\title{Diamond Dicing}

\begin{abstract}
In OLAP, analysts often select an interesting sample of the data. For example, an analyst might focus on products bringing revenues of at least  \$100\,000, or on shops having sales greater than \$400\,000. However, current systems do not allow the application of 
both of these thresholds simultaneously, selecting products and shops satisfying both thresholds. 
 For such purposes, 
 we introduce the \diamondcube\ cube operator, filling a gap among existing
   data warehouse operations.  

 Because of the interaction between dimensions  the computation of \diamondcube\ cubes is challenging.
 We compare and test various algorithms on  large data sets of
 more than 100 million facts. We find that while it is possible to implement 
 \diamondcube{}s in SQL, it is inefficient. Indeed, our custom implementation can be a hundred times faster than  popular database engines (including a row-store and a column-store).   
 
\end{abstract}


\end{frontmatter}

\section{Introduction\label{sec:intro}}   

An analyst often wants  
to focus on an \emph{interesting} part of 
her
data
set. Sometimes  this means 
she wants 
to focus on only some attribute values. For example,
she 
might select only the data related to the cities of Montreal and
Toronto between the months of July and October. This operation is a
\emph{dice} (Section~\ref{sec:formal-model}). Unfortunately, dicing requires
that the analyst know 

exactly which attribute values she 
needs. Instead of specifying the attribute values, the analyst might prefer to specify a threshold.
For example, she can make an \emph{iceberg query} (Section~\ref{sec:subsampdatabase})~: e.g.,
the cities responsible for at least \$10~million in sales. 

Unfortunately, it is difficult to apply thresholds over several dimensions.
The analyst might have selected 
cities generating at
least a certain volume of sales (\$10~million), and then 
select 
products responsible for a certain sales volume (say
\$5~million) in these cities.
Unfortunately, after selecting the popular products (\$5~million), the constraint on cities (\$10~million) may no longer be satisfied.
Moreover, the analyst could equally start from a
product selection that generates a sales volume of at least  
\$5~million, and then ask which cities have 
sales of at least 
\$10~million when considering only these products. This could produce a different result.

Instead, we propose \diamondcube{} dicing. It 
 applies  constraints
simultaneously on several dimensions in a consistent manner. For example, we may seek the cities with a sales volume
of at least \$10~million dollars, and products with a sales volume of at
least \$5~million. We require \emph{both} constraints to be \emph{simultaneously} 
satisfied.
Intuitively, \diamondcube{} dicing is a multidimensional generalisation  of icebergs. 
It 
is also an instance of dicing, but one where the analyst 
need not manually specify
 the interesting attribute values: instead, 
as with 
an iceberg query, the analyst might only specify interesting thresholds (on sales, quantities and so on). 

\begin{figure}
\centering

\begin{tabular}{!{\color{black}\vline\hspace{1pt}}c!{\color{black}\vline\hspace{-1pt}}ccccc|r|}\hline\\[-1.25em]
                                              & Chicago  & \cellcolor{blueish}{Montreal} & \cellcolor{blueish}{Miami}  &\cellcolor{blueish}{Paris} & Berlin & Totals \\\\[-1.25em]\hline
 TV                                     &   3.4   & 0.9                       &   0.1                      & 0.9                        &   2.0  & 7.3\\
\cellcolor{reddish}{Camcorder} &   0.1   & \cellcolor{grayish}{1.4}  &   \cellcolor{grayish}{3.1} &   \cellcolor{grayish}{2.3} &   2.1  & 9.0\\
\cellcolor{reddish}{Phone}     &   0.2   & \cellcolor{grayish}{6.4}  &   \cellcolor{grayish}{2.1} &   \cellcolor{grayish}{3.5} &   0.1  & 12.3\\
\cellcolor{reddish}{Camera}    &   0.4   & \cellcolor{grayish}{2.7}  &   \cellcolor{grayish}{5.3} &   \cellcolor{grayish}{4.6} &   3.5  & 16.5\\
 Game Console                            &   3.2   & 0.3                       &   0.3                      &   2.1                      &   1.5  & 7.4\\
DVD Player                              &   0.2   & 0.5                       &   0.5                      &   2.2                      &   2.3  & 5.7\\\hline
\hspace*{-4em}Totals  & 7.5 & 12.2 & 11.4 & 15.6 & 11.5 & 58.2\\\hline
\end{tabular}
\caption{Sales (in million dollars): the shaded region is a \emph{\diamondcube{}} where stores in selected cities need to
have sales above \$10~million whereas products need sales above
\$5~million. \label{tb:motivation}}

\end{figure}

Unlike regular dicing or iceberg queries, the computation of a
\diamondcube{} dice (henceforth called a \diamondcube{})
 is a challenge because of the interaction between the dimensions. 
Indeed, consider Fig.~\ref{tb:motivation}. Applying a threshold of \$10~million on sales for the cities would eliminate Chicago,
whereas applying the \$5-million threshold on products would not terminate any product.
However, once the shops in Chicago are closed,  the
products TV and Game Console  fall below the threshold of
\$5-million\footnote{The sum of TV sales is now 3.9, and the sum of Game Console sales is 4.2.}. 
We cannot stop now, after processing
each dimension once: removal of these products
causes the removal of the Berlin
store and, finally, the termination of the DVD Player product-line. 
Thus, 
simultaneously satisfying 
constraints on several dimensions may require several iterations.

We must also provide guidance regarding the
selection of the thresholds. In our example based on
Fig.~\ref{tb:motivation}, we used two thresholds (\$10~million for
stores and \$5~million for products)---but what 
 if the analyst does not have specific thresholds in mind? As a sensible default, we might put the same threshold
 $k$ on both stores and products. If $k$ is too high, the \diamondcube{} is empty. 
So we might seek $\kappa$, which is the largest value of $k$ so that the \diamondcube{} is not empty. 
 This
value $\kappa$ could be an interesting default threshold for the
analyst. In our example, $\kappa=7.4$ and the corresponding
\diamondcube{} comprises the attribute values 
Phone, Camera, Montreal, Miami, and Paris. 
Within this dice, all cities and
products have at least \$7.4~million in sales. We present and test
efficient algorithms for finding $\kappa$ (starting in
Section~\ref{sec:APrioriBoundsontheCarats}).

Our next section presents several motivating examples. Then
we present formal definitions in Section~\ref{sec:properties}. In particular, we show that our definition of a \diamondcube{} is sound by proving that there is a unique solution to the \diamondcube{} query. 
In Section~\ref{sec:algorithms}, we present 
 efficient algorithms to compute \diamondcube{}s.
We review experimentally the efficiency of our algorithms in Section~\ref{sec:experiments}. Finally, we review
related work.\footnote{Our work extends a conference paper~\cite{diamond-ideas}
 where a single algorithm was tested over small data sets.}

\section{Motivating Examples}
\label{sec:examples}

We consider example applications  to further motivate
diamond dicing.  We show how diamonds allowed us to find facts  that  surprised us in  different applications.

\paragraph{Bibliometrics example}
Consider a bibliographic table with columns
for author and venue. 
 Perhaps we
want to analyse the publication habits of professors, but much work
would be required to identify precisely which authors are professors.  
However, perhaps we can assume
that most authors without at least 5~publications, in venues where professors publish, are not professors.  
The diamond with a threshold of 5  publications per author and a threshold of one publication (from these authors) per venue  
will exclude them.  This diamond is the largest author-venue subcube 
where authors have 5~publications each in 
selected venues, and where selected venues each have at least one~publication from
selected authors. 

For illustration, we processed conference publication data 
available from DBLP~\cite{DBLPXML}; the data and details of its
preparation are given elsewhere~\cite{hazel-website}.  
See Table~\ref{tab:dblp} for some characteristics of diamonds in this data.
We find that the diamond corresponding to ``professors'' prunes about 82\% of
all authors (115\,341 out of  640\,674). Maybe surprisingly, it 
only prunes 4 venues\footnote{If we require that each author published in at least 5 \emph{different} venues, then we prune about 86\% of authors, and only 5 venues.}.
A similar result remains true if we compute the diamond corresponding to prolific ``professors'' having published at least 50~papers:
out of 5\,065, only 249 venues are pruned. Yet this diamond contains only 4\,790~authors out of 640\,674~possible
authors, which is a  selective group (less than 1\%).
Setting high thresholds is particularly useful in obtaining
smaller, more easily analysed, sets of data.  
For these purposes, we built an interactive tool that finds
the highest thresholds generating non-empty diamonds. For example, we may 
query for the largest value of $\kappa$ such that 
the following \diamondcube\  is not empty:
authors with at least $\kappa$~publications each in
retained venues, and retained venues each with at least $\kappa$~publications from
retained authors. 
In this case, the answer is $\kappa=119$. 
We found this occurrence surprising. This diamond contains 11~prolific authors in the
area of digital hardware and computer-aided design, who publish in 7~venues. 

\begin{table}\centering
{
\caption{\label{tab:dblp}Characteristics of selected diamonds from DBLP.
}
\begin{tabular}{ccc|rrcc}\hline
\multicolumn{3}{c|}{Threshold needed for each} & \multicolumn{2}{c}{Retained}   & \% size    & \\
author & \multicolumn{2}{c|}{venue} & authors & venues &  reduction & interpretation\\ \hline
1 & \multicolumn{2}{c|}{1} & 640\,673 & 5065 & 0 & all\\ 
5 & \multicolumn{2}{c|}{1} & 115\,341 & 5061 & 42 & professors \\  
50 & \multicolumn{2}{c|}{1} & 4790 & 4816 & 90 & prolific professors \\  
119 & \multicolumn{2}{c|}{119} & 11  & 7      & $>$99.9 &  hardware cluster 
\\\hline
\end{tabular}
}
\end{table}

In a modified form of the bibliometrics cube, we associated each publication
with a main keyword, obtaining a 3-dimensional cube~\cite{springerlink:10.1007/978-3-642-20095-3_47}.  Putting a
threshold on the keyword dimension can restrict analysis to popular
or mainstream topics.

Consider constraining the authors to have at least
108~publications (on mainstream topics, in popular venues), the
topics to have at least 6~occurrences (by prolific authors, in popular
venues), and the venues to have at least 20~publications (by prolific
authors, on mainstream topics).   We find
the publications
of 
I.~Pomeranz and S.~M.~Reddy in 8~hardware venues. Within these publications
the most frequent keyword is `synchronous', which  occurred 7 times more often
than the least frequent, `sequential'. Globally 
these keywords are almost equally frequent and are ranked 286th and
289th.  
These two authors are ranked 20th and 11th, globally.

\paragraph{Netflix example}
In the Netflix movie-rating database (discussed later; see 
Fig.~\ref{fig:model-map}), users have provided ratings for 
various movies, and the dates of ratings are also recorded.
Someone studying patterns in collaborative work might be interested
that there is a subset of the Netflix data where each user entered
at least 1004~ratings on movies rated at least 1004~times by these
 same users during days where there were at least 1004~ratings by 
these same users on these same movies.  We found the result surprising. 

\paragraph{Star Schema Benchmark example}

We might be interested in seeking the subset of customers and
suppliers such that each customer accounts for a sizable revenue with
selected suppliers and the suppliers each account for a sizable
revenue on those customers.  We took the fact table from the Star
Schema Benchmark~\cite{DBGEN} and rolled it up to two columns, 
customer and supplier, with revenue as the measure.  (Cube \ssbCubeOne\
statistics are given later.)  We found that about 10\% of the
customers (2\,174) each generate revenue of at least \$1.5 billion\footnote{$\kappa$ = 1\,581\,756\,429.}
 from a group of 
1\,996 suppliers (99.8\%) 
and,
simultaneously, each 
of these 1\,996 suppliers 
generates at least \$1.5 billion from
the 2\,174 
customers.  These customers and suppliers together
account for approximately 17\% of the total revenue and 16\% of the
data.  Since the Star Schema Benchmark is synthetic data generated
from  uniform distributions~\cite{tpchdescription,ssbdescription}, 
 this result is not surprising.
\section{Properties of \Diamondcube\ Cubes }
\label{sec:properties}

In this section,
 we present a formal model of the \diamondcube\ cube.
We show that \diamondcube s are nested,
with a smaller
\diamondcube\ existing within a larger \diamondcube{}.  We also prove a
uniqueness property for \diamondcube s and we establish upper and
lower bounds on the parameter $\kappa$ 
 for both  \textsc{count}  and
\textsc{sum}-based \diamondcube\ cubes.

\subsection{Formal Model\label{sec:formal-model}}

\begin{figure}
\subfloat[3-D cube]{\label{subfig:cube}
\includegraphics[height=3.5cm]{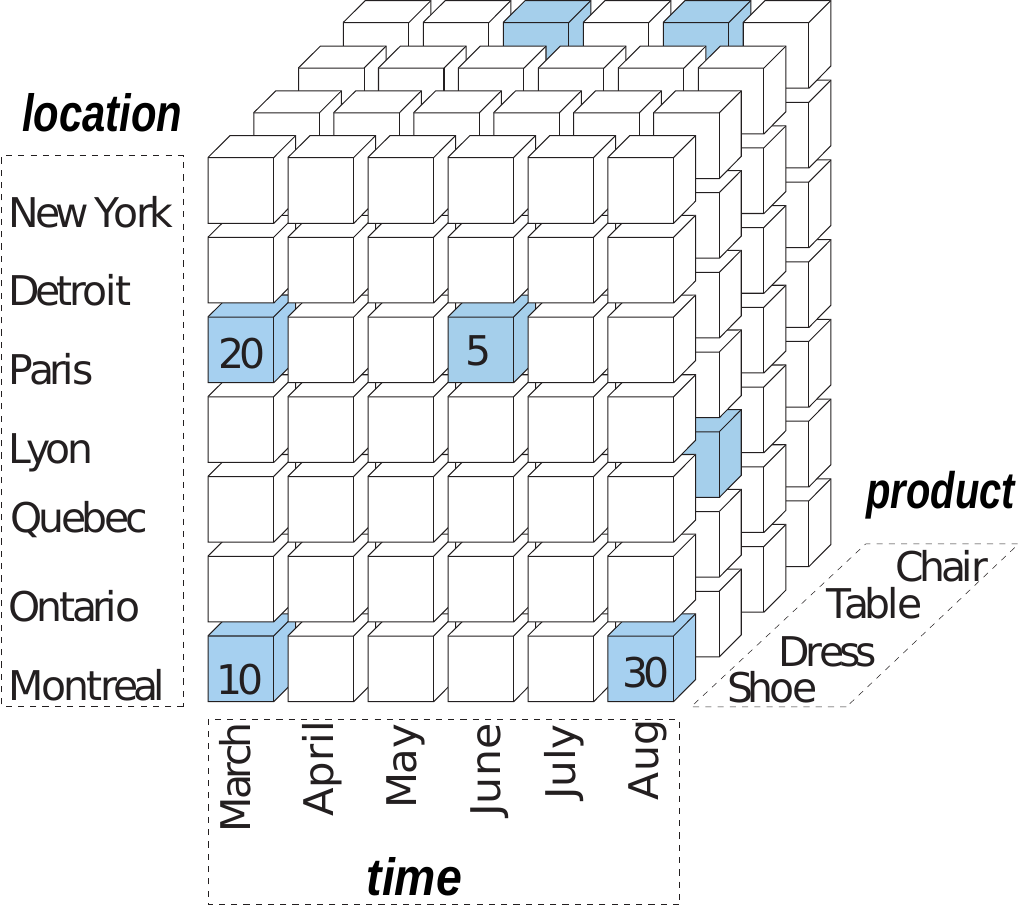}}
\subfloat[Slice on product shoe\label{subfig:slice}]{%
\includegraphics[height=3.5cm]{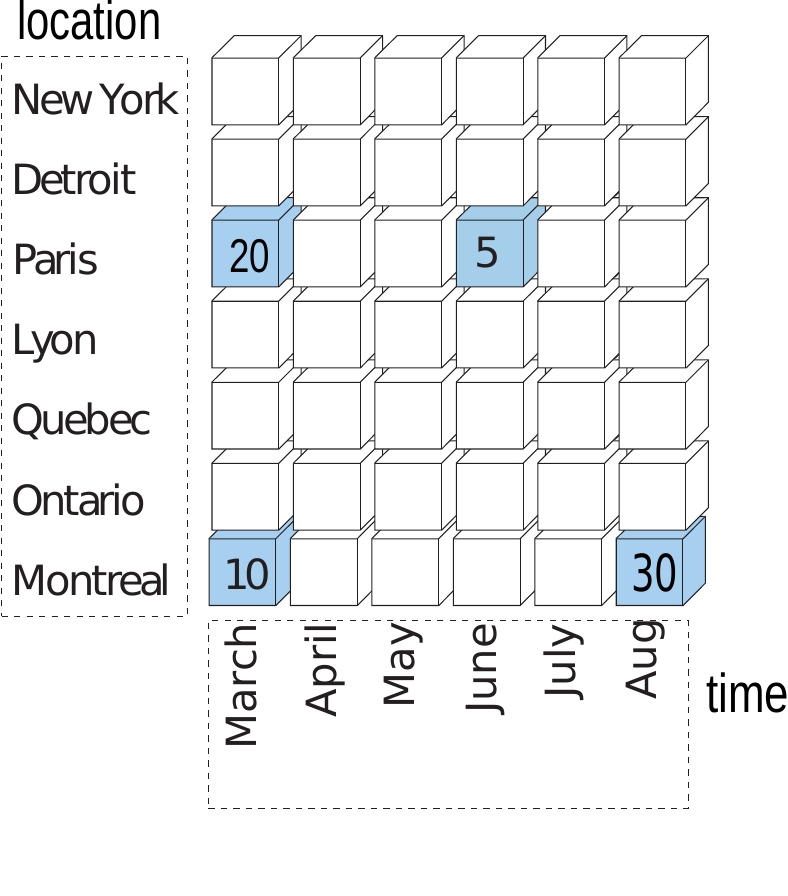}}
\subfloat[Dice on months March--May\label{subfig:dice}]{%
\hspace{0.3cm}\includegraphics[height=3.5cm]{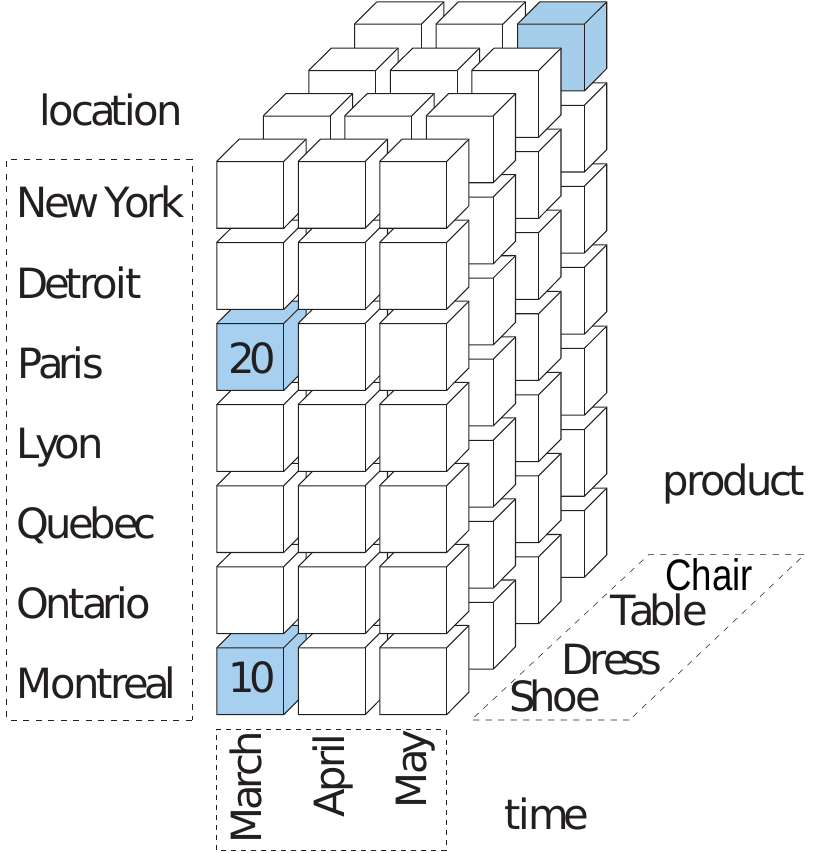}}
\caption{OLAP terms: cube, 
 slice and dice.\label{fig:cube-operators}} 
\end{figure}

Researchers and developers have yet to agree on a single multidimensional model for OLAP~\cite{1183515,Mazon2009} Our simplified formal model incorporates several widely accepted definitions for the terms illustrated in  Fig.~\ref{fig:cube-operators}, together with new terms associated specifically with \diamondcube s. For clarity, all terms are defined in the following paragraphs.

 A  \emph{dimension} $D$  is a set of \emph{attributes} that defines 
 one axis of a multidimensional data structure.  For example, in Fig.~\ref{fig:cube-operators} the dimensions are location, time and product. Each dimension $D_i$ has a cardinality $n_i$, the number of distinct attribute values in this dimension. Without losing generality, we  assume that $n_1 \leq
n_2 \leq \ldots \leq n_d$.  A dimension can be formed from a single
attribute of a database relation,  and the number of dimensions   is
denoted by $d$. 

A \emph{cube} is the 2-tuple ($\mathcal D,f$) which is 
 the set of dimensions \{$D_1, D_2, \ldots, D_d$\} together with a total function ($f$) which maps  tuples in $D_1 \times D_2 \times \dots \times D_d $ to $\mathbb{R} \cup \{\bot\}$, where $\bot$ represents undefined. Fig.~\ref{subfig:cube} shows a {cube}  with three dimensions.
 
A  \emph{cell} of cube $C$  is a 2-tuple (($x_1,x_2,\ldots,x_d) \in
D_1 \times D_2 \times \dots \times D_d,v$) where $v =
f(x_1,x_2,\ldots,x_d)$ is called a \emph{measure}.  The
measure  may be a value $v \in \mathbb{R}$, in which case we say the cell is an \emph{allocated cell}.  Otherwise, the measure is $\bot$ and we say the cell is empty---an \emph{unallocated cell}.   For the purposes of this paper, a measure is a single value.  In more general OLAP applications, a cube may map to several measures.  Also, measures  may take values other than real-valued numbers---Booleans, for example.

A \emph{slice}   is the cube $C' = (\mathcal D',f')$  obtained when  a single
attribute value is fixed in one dimension  of cube $C = (\mathcal D,f)$. 
For example, Fig.~\ref{subfig:slice} is a slice of the cube presented in Fig.~\ref{subfig:cube}.

A  \emph{dice}  defines a cube $S$ from an existing cube by removing attribute values and the corresponding cells.
For example, Fig.~\ref{subfig:dice} illustrates a dice applied to the cube from Fig.~\ref{subfig:cube} where all
months except March, April and May were removed. The resulting cube still has the same number of dimensions.
We call it a \emph{subcube} because its dimensions are subsets of the dimensions of the original cube,
and, as a function,  it is a restriction to the corresponding subset of cells.

An \emph{aggregator} is a function, $\sigma$, that assigns a real number to a set of cells---such as a slice.  For example, \textsc{sum} is an  aggregator: \textsc{sum}($\mathrm{slice}_i$) = $v_1 + v_2 +\dots + v_m$  where $m$ is the number of allocated  cells in $\mathrm{slice}_i$ and the $v_i$'s are the measures.  

 A slice $S'$ is a subset of slice $S$ if  every allocated  cell in
 $S'$  is also an allocated  cell in  $S$. An aggregator  $\sigma$ is
 monotonically non-decreasing if $S' \subset S$ implies $\sigma(S')
 \leq \sigma(S)$. Similarly, $\sigma$ is monotonically non-increasing
 if $S' \subset S$ implies $\sigma(S') \geq \sigma(S)$. Monotonically
 non-decreasing operators include  \textsc{count}, \textsc{max} and
 \textsc{sum} over non-negative measures. Monotonically non-increasing operators include  \textsc{min} and  \textsc{sum} over non-positive measures.
\textsc{mean} and  \textsc{median} are neither monotonically non-increasing, nor non-decreasing functions.

 Our formal model maps to the relational model in the following ways: (See Fig.~\ref{fig:model-map}.) 

\begin{itemize}
\item A \textbf{cube} corresponds to a  fact table: a relation whose attributes comprise a primary key and a single measure.
\item An \textbf{allocated cell} is a fact, i.e.\  it is a distinct record in a fact table.

\item A \textbf{dimension} is one of the attributes that compose the primary key.
 \end{itemize} 

\begin{figure}
\centering
\begin{tabular}{crcc}\hline
Movie & Reviewer & Date   & Rating\\\hline
1 & 1488844 & 2005-09-06  & 3\\
1 & 822109  & 2005-05-13  & 5\\
1 & 885013  & 2005-10-19  & 4\\
1 & 30878   & 2005-12-26  & 4\\
1 & 823519  & 2004-05-03  & 3\\
1 & 893988  & 2005-11-17  & 3\\
1 & 124105  & 2004-08-05  & 4\\
1 & 1248029 & 2004-04-22  & 3\\\hline
\end{tabular}
\caption[Part of the Netflix~\cite{netflixprize} fact table (cube).]{Part of the 
\nfCubeTwo~fact table (see Section~\ref{sec:real}).  Attributes (dimensions) are Movie, Reviewer and Date.  Each row is a fact (allocated cell). The measure is Rating.\label{fig:model-map}}
\end{figure}

\subsection{Diamond Cubes are Unique}\label{sec:unique}

Intuitively, a diamond cube is a subcube where all attribute values satisfy a threshold condition. For example,
all selected stores must have total sales over one million dollars. We
call such  threshold 
conditions carats.

\begin{definition}\label{def:diamond}

Given a number $k$, a  cube has $k$~carats along a dimension 
if the aggregate of every slice along that dimension is at least $k$. That
is, for every slice $x$, 
 we have $\sigma(x)\geq k$.

\end{definition}

Note that if a dimension has $k$ carats, it necessarily has $k'$ carats
for $k'  <  k $. 

Given two subcubes $A$ and $B$ of the same starting cube, their union  $A \cup B$
is defined by the union of the pairs of dimensions.
For example, if $A$ is the result of a dice limiting the location to Montreal
and $B$ is the result of a dice limiting the location to Toronto, the
subcube $A \cup B$ will be the result of a dice limiting the location
to both Montreal and Toronto.
Similarly, the intersection ($A \cap B$) is defined by the intersection
of the pairs of dimensions.
We say that subcube $A$ is contained in subcube $B$ if all of the dimensions of $A$ are contained in the
corresponding dimensions of $B$.

For  monotonically non-decreasing operators (e.g., \textsc{count}, \textsc{max} or \textsc{sum} over non-negative measures), 
union preserves the carat, as the next proposition shows.

\begin{proposition}\label{prop:union}If the  aggregator $\sigma$
is monotonically non-decreasing, then the union of any two cubes having
  $k_i$ (resp. $k'_i$) carats along dimension $D_i$  has $\min(k_i,k'_i)$ carats
  along dimension $D_i$ as well, for  $i = \{1,2,\dots,d\}$. 
 
\end{proposition}

\begin{proof}
The proof follows from the monotonicity of the aggregator. \end{proof}

If we limit ourselves to monotonically non-decreasing
aggregators, then we can efficiently seek the largest possible subcube satisfying a given set of carats.
We call such a subcube the \textbf{diamond}.

\begin{definition}
The $k_1, k_2, \ldots, k_d$-carat diamond is the maximal 
subcube having\\ $k_1, k_2, \ldots, k_d$~carats 
along dimensions $D_1,D_2\ldots, D_d$. 
That is, any subcube having $k_1, k_2, \ldots, k_d$ carats is contained in the diamond.
\end{definition}

By Proposition~\ref{prop:union}, the diamond is unique when $\sigma$
is  monotonically non-decreasing: it is given by the union of all subcubes having $k_1, k_2, \ldots k_d$ carats. For more general aggregators or when different aggregators are applied
to different dimensions, the computation of the diamond might be
NP-hard or ill-defined.  For instance, when \textsc{sum} is used over cubes having both
positive and negative measures,
there may no longer be a \emph{unique} solution to the problem `find the
$k_1, k_2 \dots k_d$-carat cube'.  This is indeed the case for the
cube in Fig.~\ref{fig:reinstate-attr}.  

\begin{figure}
\centering
\subfloat[Cube with positive and negative measures.]{%
\begin{tabular}{l|cccccc}\hline
row & col 1 & col 2 & col 3 & col 4 & col 5 & col 6\\\hline
1   & -5 & 1 & 1 & 1 & 0 & 3\\
2   & -3 & -4 & 1 & 0 & 1 & 0\\
3   &  2 &  2 & 4 & 0 & 2 & 1\\
4   &0 & 2 & 3 & 1 & 0 & 0\\\hline
\end{tabular}
}

\subfloat[\textbf{rows} processed first.]{%
\begin{tabular}{l|cc}\hline
row  & col 2 & col 3 \\\hline
3    &  2    & 4     \\
4    & 2     & 3     \\\hline
\end{tabular}
}
\qquad
\subfloat[\textbf{columns} processed first.]{%
\begin{tabular}{l|cc}\hline
row & col 3 & col 6\\\hline
1      & 1     & 2\\
4      & 4     & 2\\\hline
\end{tabular}
}

\caption{\label{fig:reinstate-attr}There is no unique 4,4-carat \textsc{sum}-based \diamondcube.  
}
\end{figure}

Sometimes we require the same carat $k$ along all dimensions. To simplify the notation, instead of writing ``$k, k, \ldots, k$-carat'', we  write ``$k$-carat''.

\subsection{A Priori Bounds on the Carats}
\label{sec:APrioriBoundsontheCarats}

The computation of a \diamondcube\ requires that the analyst specify the
desired number of carats. However, this may not be practical for
all dimensions. For example, the analyst may want to select
stores with sales above one million dollars, but she may not
know how to select the threshold for the product dimension. In 
such cases, it might be best to set the carats to the largest
possible value that generates a non-empty diamond.
This maximal number of carats can be found efficiently by binary search
if we can determine a limited range of possible values.

Given a cube $C$ and $\sigma$, then  $\kappa$  is the largest number of
carats for which $C$  
has a non-empty \diamondcube. 
Intuitively, a small cube with many allocated cells should have a large $\kappa$, and the following
proposition makes this precise.

\begin{proposition}For \textsc{count}-based carats, we have
\label{prop:kappa}
$\kappa \geq (|C|-1)/\sum_{i=1}^d (n_i-1)$.
\end{proposition}
\begin{proof}
We begin by proving that, for \textsc{count}-based carats,
if a cube  $C$   does not contain
a non-empty $k$-carat  subcube, then \begin{eqnarray}|C|\leq 1+(k-1)\sum_{i=1}^d (n_i-1).\label{eqn:oldthm}\end{eqnarray}
Suppose that a  cube $C$ of dimension at most $n_1 \times n_2 \times\dots \times n_d$ contains no $k$-carat \diamondcube. Then one slice must
contain at most $k-1$~allocated cells. Remove this slice.
The amputated cube must not contain a $k$-carat \diamondcube.
Hence, it has one slice containing at most $k-1$~allocated cells.
Remove it. This iterative process can continue at most
$\sum_i (n_i-1)$ times before there is at most one allocated cell left:
hence, there are at most $(k-1)\sum_i (n_i-1) + 1$~allocated cells in total.

By definition of $\kappa$, we have that the cube does not contain 
a non-empty $\kappa+1$-carat  subcube. By substitution ($k\rightarrow \kappa+1$) in Equation~\ref{eqn:oldthm}, we have that
$|C|\leq 1+\kappa \sum_{i=1}^d (n_i-1)$. Solving for $\kappa$,
we have $\kappa \geq (|C|-1)/\sum_{i=1}^d (n_i-1)$.
\end{proof}

Based on this lower bound alone, we  compute $\kappa$ efficiently  (see Section~\ref{sec:findingkappa}). For a related discussion on \textsc{sum}-based diamonds, see~\ref{sec:sum-diamond-properties}.

\section{Algorithms\label{sec:algorithms}}

Computing \diamondcube s is challenging because of the interaction
between dimensions; modifications to a measure associated with an
attribute value in one dimension have a cascading effect through the
other dimensions. We use  different approaches to compute \diamondcube s:
\begin{itemize}
\item  
We implemented a custom program in Java   
that loops through the cube checking and updating
the \textsc{count} or \textsc{sum} for all attribute values in each dimension until it stabilises (see Section~\ref{sec:external}). 
\item  We also implemented an algorithm using SQL\@ (see Section~\ref{sec:sqlalgo}). 
\end{itemize}

We based both our custom and SQL implementations on the basic  algorithm for computing \diamondcube s given in Algorithm~\ref{algo:basic}.  Its  approach is to repeatedly identify an attribute
value that cannot be in the \diamondcube, and then remove the
attribute value and its slice. The identification of ``bad'' attribute
values is done conservatively, in that they are known already to have
a sum less than required ($\sigma$ is \textsc{sum}), or insufficient
allocated cells ($\sigma$ is \textsc{count}).   When the algorithm
terminates, only attribute values that meet the condition in every slice remain: a \diamondcube.

\begin{algorithm}\algosize{}
\SetAlgoRefName{BASIC} 

\dontprintsemicolon
 \SetKwInOut{Input}{input}\Input{a $d-$dimensional data cube $C$, a  monotonically non-decreasing  aggregator $\sigma$ and  $k_1\geq 0, k_2\geq 0, \ldots, k_d \geq 0$}
 \SetKwInOut{Output}{output}\Output{the diamond cube $A$}
\SetLine
      $\varfont{stable} \leftarrow  \constfont{false}$\;

 \While {$\neg \varfont{stable}$ }{%
      $\varfont{stable} \leftarrow  \constfont{true}$\;
    \tcp{major iteration}
	     \For {$\varfont{dim} =1,\ldots,d$}{%

        \For {$\varfont{i}$ 
        in all attribute values of dimension $\varfont{dim}$}{%
            $C_{\varfont{dim},\varfont{i}}\leftarrow \sigma(\text{slice }$i$\text{ on dimension }\varfont{dim}) $\;
            \If{$C_{\varfont{dim},\varfont{i}} < k_\varfont{dim}$}{%
                delete attribute value $\varfont{i}$\;
                \If{$C_{\varfont{dim},\varfont{i}} >0$ and $\varfont{dim}>1$ }{$\varfont{stable} \leftarrow \constfont{false}$}\;
            }
        }
     }
   }
    
   return cube without deleted attribute values;
   \vspace{1em}
\caption{\label{algo:basic}Algorithm to compute the \diamondcube\ of
any given cube by deleting slices eagerly. }
\end{algorithm}

Algorithms based on this approach  always terminate, though they might
sometimes return an empty cube. By specifying \emph{how} to  compute
and maintain counts (or  sums) for each attribute value in every dimension  we  obtain different variations. The correctness of any such variation
is guaranteed by the following result.

\begin{theorem}\label{thm:algoiscorrect}Algorithm~\ref{algo:basic} is correct, that is, it always returns
the $k_1,k_2, \ldots, k_d$-carat \diamondcube.
\end{theorem}

\begin{proof}
Because the \diamondcube\ is unique, we need only show that the
result of the algorithm, the cube $A$, is a \diamondcube.
If the result is not the empty cube, then dimension $D_i$ has at least
 aggregated value
$k_i$  per slice, and hence the cube 
 has $k_i$ carats.   (Note that, in the
main loop, if only attribute values having zero aggregates 
 are deleted in all but the first dimension, it is not necessary to do another pass.) We only need
to show that the result of Algorithm~\ref{algo:basic} is maximal:
there does not exist a larger $k_1, k_2, \ldots, k_d$-carat cube.

Suppose $A'$ is such a larger $k_1, k_2, \ldots, k_d$-carat cube.  Because
Algorithm~\ref{algo:basic} begins with the whole cube $C$, there
must be a first time when one of the attribute values
of dimension $\varfont{dim} \textrm{ of } C$
 belonging to $A'$ but not $A$ is deleted.

At the time of deletion, the slice corresponding to this attribute value had aggregate measure less than $k_{\varfont{dim}}$.  
Let $C'$ be the cube at the instant before the attribute is deleted,
with all attribute values deleted so far.
We see that $C'$ is larger than or equal to $A'$, a $k_1, k_2, \ldots,
k_d$-carat cube 

  and therefore,
slices in $C'$ corresponding to attribute values of $A'$ along
dimension $\varfont{dim}$ must have aggregate measures of at least
$k_{\varfont{dim}}$, corresponding to   $k_{\varfont{dim}}$~carats 
 Therefore, we have a contradiction and must
conclude that $A'$ does not exist and that $A$ is maximal.
\end{proof}

If the aggregator is \textsc{count}, and $k_i=1$ for $i>1$, then Algorithm~\ref{algo:basic} computes the \diamondcube{} in a single pass. 

 \subsection{Custom Software}
\label{sec:external}
The size of available memory affects the capacity of in-memory data
structures to represent data cubes. In our experiments, we used 
a typical  laptop computer with 8\,GiB 
 of memory.  When we restricted the amount of memory available\footnote{We used an \texttt{mlock} system 
call to remove 6\,GiB of memory from use on our 8\,GiB computer.}, all execution times slowed, but our custom Java software still out-performed the database management systems. 

 We are interested
in processing large data.  Therefore, we seek an  efficient external
memory implementation where the data cube can be stored in an external file whilst the important \textsc{count}s (or \textsc{sum}s) are maintained in memory.  

Algorithm~\ref{algo:basic} checks the $\sigma$-value for each attribute on
every iteration. Calculating this value directly, from a data cube too
large to store in main memory, would
entail many expensive disk accesses.  Even with the \textsc{count}s
maintained in main memory, it is prudent to reduce the number of I/O
operations as much as possible.  One way this can be achieved is to
store the data cube as normalised binary integers  using bit
compaction~\cite{ng1997block}---mapping strings to small integers
starting at zero. 

 Algorithm~\ref{algo:shrinkingInput} (External-Memory-Diamond-builder) 
 employs $d$~arrays, $\varfont{a}_1$ to $\varfont{a}_d$, that map attributes to their aggregate $\sigma$-values.  As values are pruned from the diamond, we must repeatedly update these arrays so that they continue to
maintain the aggregate of each slice.
This update  can be executed in constant time for aggregators such as
\textsc{count} and \textsc{sum}: in the notation of
Algorithm~\ref{algo:shrinkingInput}, the update
is computed as $\varfont{a}_j(x_j) = \varfont{a}_j(x_j) - \sigma(\{r\})$. 

\begin{algorithm}\algosize{}
\dontprintsemicolon

\SetAlgoRefName{EMD} 

 \SetKwInOut{Input}{input} \Input{file $\varfont{inFile}$ containing
   $d-$dimensional cube $C$,\\ integers $k_1, k_2, \ldots, k_d > 0$,\\
  monotonically non-decreasing  aggregator $\sigma$,\\
a parameter $\tau\in [0,1)$ (we used $\tau = 0.5$)
   }
   
 \SetKwInOut{Output}{output}\Output{the cells in the diamond data cube}
 \SetLine

\ForEach{dimension $i$} {%
   Create array  $\varfont{a}_i$ 
   of size $|D_i|$\;
   \ForEach {attribute value $v$ in dimension $i$} {%

     $\varfont{a}_i(v) = \sigma($ slice for value $v$ of dimension $i$ in $C)$
   }
}
 
 $\varfont{stable} \leftarrow \constfont{false}$\;
 \While {$\neg \varfont{stable}$} {%
   $\varfont{stable} \leftarrow \constfont{true}$\;
   \ForEach{cell $r$ of $\varfont{inFile}$ which is not marked as deleted}{%
     $((x_1, x_2, \ldots, x_{d}),v) \leftarrow r$\;
     \For{$i \in \{1, \ldots, d\}$}{%
	       \If{$\varfont{a}_i (x_i) < k_i$}{%
\tcp{attribute $x_i$ had previously been deleted}
         \For{$j \in \{1,\ldots, i-1, i+1, \ldots, d\}$ }{%
         	update $\varfont{a}_j(x_j)$ given the removal of cell~$r$ \;
         }
         $\varfont{stable} \leftarrow \constfont{false}$\;
         mark $r$ as deleted\;
         $\constfont{break}$\;\tcp{only delete this cell once}
	       }
       }

     }
     \If{the fraction of cells marked as deleted exceeds $\tau$}{%
     	rebuild $\varfont{inFile}$ without the deleted cells\;
     }
   }
     	rebuild $\varfont{inFile}$ without the deleted cells\;
 \Return{$\varfont{inFile}$
 }
    \vspace{1em}

\caption{\label{algo:shrinkingInput}
 \Diamondcube\ dicing for relationally stored cubes. With each iteration,
 less data is processed.}
\end{algorithm}

Each time the algorithm passes through the data, it updates the aggregates eagerly, 
and marks cells as deleted. Only when a significant fraction of the cells
have been marked as such, 
are the cells actually deleted: 
we found it efficient to rebuild the list of cells when more than half
 have been marked as deleted ($\tau=0.5$).

When memory is abundant, we can use
Algorithm~\ref{algo:shrinkingInput} 
 while keeping
the content of the files in memory.  Indeed, such a version was
implemented and we refer to this
algorithm as \imd\  (In-Memory-Diamond-builder)
in Section~\ref{sec:experiments}.

\subsection{An SQL-Based Implementation\label{subsec:sql}}
\label{sec:sqlalgo}

\newsavebox{\tempbox}%

\begin{figure}
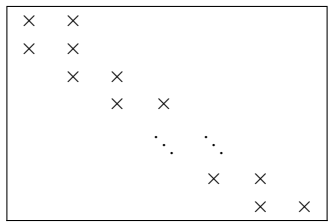

\centering \footnotesize
$\begin{array}[c]{|ccccccc|}\hline
\times & \times &   &   &   &   &  \\
\times & \times &   &   &   &   &  \\
  & \times & \times &   &   &   &  \\
  &   & \times & \times &   &   &  \\
  &   &   & \ddots & \ddots &   & \\
  &   &   &   & \times & \times & \\
  &   &   &   &   & \times & \times\\\hline
\end{array}$
\caption{An $n \times n$ cube with $2n$ allocated cells (each  indicated by a $\times$) and a 2-carat diamond in the upper left:  it is a difficult case for several algorithms.\label{fig:lotsa-iterations2}}
\end{figure}

  Formulating a \diamondcube\ cube query in SQL-92 is challenging.
Using nested queries and joins, we could essentially simulate a fixed number
of iterations of the outer loop in Algorithm~\ref{algo:basic}.
Unfortunately, we do not know how to determine the number
of iterations without computing the \diamondcube{}  itself.
 Consider Fig.~\ref{fig:lotsa-iterations2} and the corresponding  2-carat \textsc{count}-based  \diamondcube{}. Using Algorithm~\ref{algo:basic}, $n-2$~iterations are required to find the \diamondcube{}.  That is, we
see an example where the number of iterations  
 $I = \Omega(n)$ and
stopping after $o(n)$ iterations results in a poor approximation
with $\Theta(n)$ allocated cells and attribute values---whereas the
true 2-carat diamond has 4~attribute values and 4~allocated cells.

\begin{algorithm}\algosize{}
\SetAlgoRefName{SQL} 
\begin{algorithmic}
 \STATE \textbf{INPUT:} a $d-$dimensional data cube $C$ and $k>0$
 \STATE \textbf{OUTPUT:} the \diamondcube\ $A$
 \STATE initialise $\mathcal{R}$ to $C$, the fact table
 \REPEAT [major iteration]
     \STATE execute the fragment of SQL pseudocode shown below
 \UNTIL {no records were deleted from $\mathcal{R}$}
  \STATE return $\mathcal{R}$ as $A$
 \end{algorithmic}

\noindent
\hrulefill

\begin{lstlisting}
CREATE TABLE $\mathit{temp}_1$ AS
  (SELECT $\mathit{dim}_1$ FROM $\mathcal{R}$
   GROUP BY $\mathit{dim}_1$ HAVING $\sigma(\mathit{measure}) < k$);

$\ldots$
CREATE TABLE $\mathit{temp}_d$ AS
  (SELECT $\mathit{dim}_d$ FROM $\mathcal{R}$
   GROUP BY $\mathit{dim}_d$ HAVING $\sigma(\mathit{measure}) < k$);
DELETE FROM $\mathcal{R}$
  WHERE $\mathit{dim}_1$   IN (SELECT * FROM $\mathit{temp}_1$) OR $\ldots$
        $\mathit{dim}_d$  IN (SELECT * FROM $\mathit{temp}_d$);
\end{lstlisting}

\caption{\label{algo:bothSQL-and-loop} Variation where the inner two loops
in Algorithm~\ref{algo:basic} are  computed in SQL\@. This process can be repeated until
$\mathcal{R}$ stabilises.}
\end{algorithm}

 We express
the essential calculation in 
SQL, as 
 Algorithm~\ref{algo:bothSQL-and-loop}. It is implemented as a stored
 procedure  in SQL:1999, which allows the iterations to be controlled entirely within the DBMS\@.
 Algorithm~\ref{algo:bothSQL-and-loop} is 
executed against a  copy of the fact table, which  becomes smaller as the algorithm progresses. The fastest
 variation of this algorithm does not delete
 slices immediately, but instead updates Boolean values to indicate
 the slices not included in the solution.  The
 data cube is rebuilt when 75\% of the remaining cells are marked
 for deletion.  B-tree indexes are built on each dimension to
 facilitate faster execution of the many GROUP BY clauses. 

\subsection{Complexity Analysis \label{subsec:complexity}}

 Algorithm~\ref{algo:basic}
visits each dimension in  sequence until it stabilises.  Ideally, the stabilisation should occur after as few iterations as possible.

Let $I$ be the number of iterations through the  input file till convergence;  i.e.\ no more deletions are done.
Value $I$ is data dependent and (by Fig.~\ref{fig:lotsa-iterations2}) is $\Theta( \sum_i n_i)$ in
the worst case.  In practise,  $I$ is not expected to be
nearly so large, and working with  large real  data sets
  $I$ did not exceed 56.  
Initial experiments suggested that the relationship of $I$ to $k$ would be
non-decreasing to $\kappa + 1$ and
non-increasing thereafter.  
Unfortunately, there are some cubes for
which this is not the case. 
  Fig.~\ref{fig:I-to-k-counterexample}
illustrates such a cube,  where $\kappa = 3$.  On the first iteration, processing columns
first for the  2-carat \diamondcube{}, a single cell is deleted.  On
subsequent iterations at most two cells are deleted until convergence.
However, the 3-carat and 4-carat \diamondcube{}s  converge
 after a single iteration. 

 The value of  $k$ relative to $\kappa$ does, however, influence  $I$.  Typically, when $k$ is far from $\kappa$---either less or greater---fewer iterations are required to converge. However, when $k$ exceeds $\kappa$ by a very small amount, say 1,  then  typically 
many more iterations are required to converge to the empty cube.  

\begin{figure}
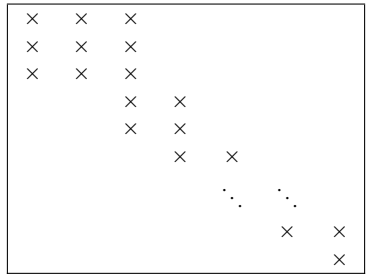

\centering\footnotesize
\begin{tabular}{|ccccccc|}\hline
$\times$  & $\times$  & $\times$ &          &           &          & \\
$\times$  & $\times$  & $\times$ &          &           &          &\\
$\times$  & $\times$  & $\times$ &          &           &          &\\
          &           & $\times$ & $\times$ &           &          &\\
          &           & $\times$ & $\times$ &           &          &\\
          &           &          & $\times$  & $\times$ &          & \\
          &           &          &          & $\ddots$  & $\ddots$  &\\

          &           &          &          &           & $\times$ & $\times$ \\
          &           &          &          &           &         & $\times$\\\hline
\end{tabular}
\caption{The 2-carat \diamondcube\ requires more iterations to
  converge than the 3-carat \diamondcube{}. 
  Allocated
  cells are indicated by a $\times$.\label{fig:I-to-k-counterexample}}
\end{figure}

Algorithm~\ref{algo:shrinkingInput} 
runs in time O($I d |C|$).
  Often, the number of
attribute values remaining in the \diamondcube\  
 decreases  substantially in the first few iterations and those cubes are  processed faster than this bound suggests.
The more carats we seek, the faster the cube decreases initially.

\section{Experiments\label{sec:experiments}}
  We show that \diamondcube s can be
computed efficiently, i.e.\ within a few minutes on 
a typical
laptop computer, even for very  large data sets.   Some of the properties
of \diamondcube s, including their size
 and the range of values the carats may take, were 
assessed experimentally.

\subsection{Hardware and Software\label{sec:hardware}}

All experiments were conducted on a Gateway NV59 notebook with dual Intel
 i5 M430 (2.27\,GHz) processors with  8\,GiB of DDR3-1066 RAM running Ubuntu 12.04. The hard disk is a 596\,GiB
ATA WDC6400BEVT-22AORTO running at  5\,400\,rpm.  It has an estimated reading speed of 86\,MB/s. 

 The algorithms were implemented in Java, using SDK version 1.7.0 and the default value (1.66\,GiB) for maximal heap size,
 and the code was  archived at a public website~\cite{hazel-website}.  Algorithm~\ref{algo:bothSQL-and-loop} was implemented in both an RDBMS~(MySQL) and a column-store DBMS~(MonetDB)~\cite{boncz2005}. 
 RDBMS experiments were conducted on MySQL version~5.5 Community Server
 with MyISAM storage engine.  
 MySQL is used in data warehousing and OLAP, most notably through 
 vendors such as  
 Infobright~\cite{Slezak:2009:DWT:1559845.1559933}, JasperSoft and Pentaho~\cite{Bouman:2009:PSB:1795840}. The column-store experiments were conducted with MonetDB~11.11.5. 

 Both database implementations make use of  stored procedures and a Java interface collected execution times.  The drivers used were  MySQL Connector/J~5.1.21 and monetdb-jdbc~2.3. 

These database systems handle index creation differently:
\begin{itemize}
\item 
 Of the index structures
 available in this version of MySQL, only B-trees are appropriate to
 the  \diamondcube\ dice operation. Spatial indexing  is limited to two
 dimensions and hash indexing requires that the data reside in main
 memory.  We built B-tree indexes on all columns to speed-up the GROUP-BY computations.

\item In MonetDB, index creation is automatically determined with no option for the user to override  system decisions~\cite{monetDB:online-docs}. Different data compression techniques, including dictionary encoding for all strings, reduce the memory footprint.

\end{itemize}

\subsection{Data Used in Experiments}
\label{sec:data}

A varied selection of freely-available real-data sets together
with some systematically generated synthetic data sets were used in
the experiments.  Each data set had a particular characteristic: a
few dimensions or many, dimensions with high or low cardinality or a
mix of the two, small  or large number of cells.  They were chosen to
illustrate that diamond dicing is tractable under varied conditions
and on many different types of data.

\subsubsection{Real Data\label{sec:real}}
\begin{table}[tb]
\centering
\caption{\label{tb:stats}Statistics of data sets.
}
\begin{tabular}{llrrrl}\hline
source                       & cube         & $d$  & $|C|$          &$\sum_{i=1}^d n_i$ & measure\\\hline
\multirow{6}{*}{King James Bible~\cite{Gutenberg}} &\bCubeOne  & 4          & 54\,601\,077     & 31\,634       & \textsc{count} \\
&\bCubeTwo   & 4         & 24\,000\,000     & 27\,042      & \textsc{count}\\
&\bCubeThree & 4         & 32\,000\,000     & 29\,078       & \textsc{count}\\
&\bCubeFour  & 4         & 40\,000\,000     & 30\,417      & \textsc{count}\\
&\bCubeFive  & 4          & 54\,601\,077     & 31\,634       &
\textsc{sum} occurrences\\
&\bCubeEight  & 10       & 365\,231\,367     & 6\,335        & \textsc{count} \\\hdashline[1pt/1pt]

\multirow{2}{*}{Census Income~\cite{MLRepository}}& \cCubeOne    & 27         & 135\,753      & 504   & \textsc{count}      \\
& \cCubeTwo    & 27         & 135\,753      & 504      & \textsc{sum}\,stocks \\\hdashline[1pt/1pt]
\multirow{3}{*}{DBLP~\cite{DBLPXML} } & \dblpCubeOne & 2    & 1\,791\,857    & 645\,739        & \textsc{count}\\
                             & \dblpCubeTwo & 2    & 1\,791\,857    & 645\,739        & \textsc{sum} publications\\
                             & \dblpCubeThree & 3  & 2\,516\,364    & 689\,589        & \textsc{count}\\\hdashline[1pt/1pt]
\
\multirow{3}{*}{Netflix~\cite{netflixprize}}
& \nfCubeOne   & 3 & 100\,478\,158  & 484\,141 & \textsc{count} \\
& \nfCubeTwo   & 3          & 100\,478\,158  & 484\,141 & \textsc{sum} rating \\
& \nfCubeThree & 4          & 20\,000\,000   & 473\,753 &
\textsc{count} \\\hdashline[1pt/1pt]
\multirow{3}{*}{Tweed~\cite{tweed}}& \twCubeOne   & 4          & 1\,957        & 91     & \textsc{count}\\
& \twCubeTwo   & 15         & 4\,963        & 674     & \textsc{count}\\
& \twCubeThree & 15         & 4\,963        & 674     & \textsc{sum}
killed\\\hdashline[1pt/1pt]

\multirow{2}{*}{Weather~\cite{hahn:weatherbench}}& \wCubeFour  & 11         & 124\,164\,371    & 48\,654       & \textsc{count}\\
& \wCubeTwo  & 11         & 124\,164\,371    & 48\,654       &
\textsc{sum}  cloud cover\\\hline
\end{tabular}
\end{table}

Five of the real-data sets  were downloaded from the following 
sources:

\begin{enumerate}
\item Census Income:  \url{http://archive.ics.uci.edu/ml}~\cite{MLRepository}
\item DBLP: \url{http://dblp.uni-trier.de/xml/}
\item Netflix: \url{http://www.netflixprize.com}~\cite{netflixprize}
\item TWEED: \url{http://folk.uib.no/sspje/tweed.htm}~\cite{tweed}
\item Weather: \url{http://cdiac.ornl.gov/ftp/ndp026b/}~\cite{hahn:weatherbench}

\end{enumerate}
 Details of how the cubes were extracted are available at a
public website~\cite{hazel-website}.  For cube \twCubeOne\ we chose
 four attributes: Year, Country, Type of Action and Target of Action with
cardinalities of 53, 16, 11 and 11, respectively. The attributes for
cubes \nfCubeOne{}, \nfCubeTwo\ and \nfCubeThree\ are Movie (17\,770),
Reviewer (480\,189) and Date (2\,182).   Rating (5) is the measure for \nfCubeTwo{}.
Their statistics are given in
Table~\ref{tb:stats}.
  Each cube was stored relationally in
a comma-separated file on disk.  A brief description of how  data cubes were
extracted from the King James Bible data follows. 

The data set was generated from the King James version of the
Bible available at Project Gutenberg~\cite{Gutenberg}.
KJV-4grams~\cite{kaserdolap2008,rlewithsorting} is a data set motivated by
applications of data warehousing to literature. It is a large list
(with duplicates) of 4-tuples of words obtained from the
verses in the King James Bible~\cite{Gutenberg}, after stemming with
the Porter algorithm~\cite{275705} and removal of stemmed words with
three or fewer letters. Occurrence of row $w_1, w_2, w_3, w_4$
indicates a verse contains words $w_1$ through $w_4$, in this order. This data is a scaled-up version of word co-occurrence cubes
used to study analogies in natural
language~\cite{TurneyML,KaserKeithLemire2006}. These data were chosen
to be representative of large cubes that might occur in text-mining applications.

  Cube 
\bCubeOne\ was extracted from KJV-4grams.  Duplicate records were removed and a count of each
unique sequence was kept, which became the  measure for cube
\bCubeFive{}. Four subcubes of 
  \bCubeOne\ were also processed:  \bCubeTwo\ has the first
  24\,000\,000~rows; \bCubeThree\ has the first 32\,000\,000~rows; and
  \bCubeFour\ has the first 40\,000\,000~rows. KJV-10grams has similar properties to KJV-4grams, except that there
 are 10~words in each row and the  process of creating KJV-10grams
 was terminated when 500~million records had been generated---at the end
 of Genesis~19:30.  Cube \bCubeEight\ was extracted from KJV-10grams.
 The statistics for all six cubes are also given in Table~\ref{tb:stats}. 

\subsubsection{Synthetic Data \label{sub:synthetic}}

We took the fact table from the Star Schema Benchmark~\cite{DBGEN} 
 and rolled-up on the supplier and customer dimensions to create cube \ssbCubeOne{}.
  The result has 2\,000  suppliers, 20\,000 customers, and
over five million rows. Uniform distributions are 
used to generate the
benchmark~\cite{tpchdescription,ssbdescription} 
and the data is 
lacking correlations between columns that real data would frequently possess.

 To investigate the effect that data distribution might have on the
 size  and shape of \diamondcube s, nine  cubes of varying dimensionality and distribution were constructed.  We chose
 1\,000\,000 cells with replacement from  each of three different
 distributions:

\begin{itemize}
\item  uniform---cubes \uCubeOne{}, \uCubeTwo{},
  \uCubeThree{}.
\item power law with exponent 3.5 to model the 65-35 skewed
  distribution---cubes \sCubeOne{}, \sCubeTwo{}, \sCubeThree{}.
\item  power law with exponent 2.0 to model the 80-20 skewed
  distribution---cubes \ssCubeOne{}, \ssCubeTwo{}, \ssCubeThree{}.
\end{itemize}

 Details of the cubes generated are given in Table~\ref{tb:synth}.

\begin{table}
\centering
\caption{Statistics of the synthetic data cubes. \label{tb:synth}}
\begin{tabular}{l|r|r|r|c}\hline
Cube         &  $d$     &$|C|$         & $\sum_i n_i$      & measure\\\hline
\ssbCubeOne  &   2      & 5\,524\,778     &22\,000         & \textsc{sum} revenue\\\hdashline[1pt/1pt]
\uCubeOne    &    3          & 999\,987      &  10\,773    & \textsc{count}\\
\uCubeTwo    &    4          & 1\,000\,000   &  14\,364    & \textsc{count}\\
\uCubeThree  &    10         & 1\,000\,000   &  35\,910    & \textsc{count}\\\hdashline[1pt/1pt]
\sCubeOne    &    3          & 939\,153      & 10\,505     & \textsc{count}\\
\sCubeTwo    &    4          & 999\,647      & 14\,296      & \textsc{count}\\
\sCubeThree  &   10          & 1\,000\,000   & 35\,616       & \textsc{count}\\\hdashline[1pt/1pt]
\ssCubeOne   &    3          & 997\,737      & 74\,276    & \textsc{count} \\
\ssCubeTwo   &    4          & 999\,995      & 99\,525     & \textsc{count}\\
\ssCubeThree &   10          & 1\,000\,000   & 248\,703    & \textsc{count}\\\hline
\end{tabular}
\end{table}

\subsection{Preprocessing Step \label{sec:preproc}}

Before applying Algorithm~\ref{algo:shrinkingInput}, we need to
convert the input (flat text files)  to flat binary files.
To determine if row ordering would have an effect on our
implementation of Algorithm~\ref{algo:shrinkingInput}, we chose two cubes---\cCubeOne\ and \bCubeTwo ---and shuffled the rows using
the GNU   utility \emph{shuf}.  We compared preprocessing and
processing times for each of six cubes, averaged over ten runs.
Extracting  cubes from the data sets included a sorting step so that
duplicates could be easily removed.   We found that preprocessing
the cube sorted on its dimension of largest cardinality was up to  25\% faster than preprocessing the shuffled
cube.  However, execution times for
Algorithm~\ref{algo:shrinkingInput} were within 3\% for each cube.
Therefore,  we did not reorder the rows prior to
processing.  We also found no significant difference in execution
times when the cubes were sorted by different dimensions.

As stated in Section~\ref{sec:external} we implemented a version (called
\imd ) of
Algorithm~\ref{algo:shrinkingInput} that reads the data cube entirely into main
memory whenever the cube is small enough ($< 1$\,GiB).
 Otherwise, the
algorithm processes the cube in a similar fashion. 

The  algorithms used in our experiments require 
different
preprocessing of the cubes. For both 
Algorithms~\ref{algo:shrinkingInput} and  \imd{}, 
 an in-memory data structure is
used to maintain aggregates of the attribute values.
Algorithm~\ref{algo:bothSQL-and-loop} references the cube stored in a database management system.   Consequently, the preprocessor writes different kinds of data to supplementary files depending on which algorithm is to be used.

\begin{table}[tb]
\centering
\caption[Preprocessing times.]{\label{tb:preproc-times}Wall-clock
  times (in seconds) for preprocessing real-world data
  sets.  A `---' indicates that this algorithm was not applied to the
  corresponding data cube.
} 
\begin{tabular}{ll|ll}\hline
             &                                &\multicolumn{2}{c}{\ref{algo:bothSQL-and-loop}}\\
Cube         &     \ref{algo:shrinkingInput}  & MySQL                      & MonetDB \\ \hline

\bCubeOne    & $2.1 \times 10^2$              & $1.1 \times 10^3$           & $7.1 \times 10^1$       \\

\bCubeTwo    & $5.7 \times 10^1$              & $5.0 \times 10^2$           & $3.1 \times 10^1$ \\
\bCubeThree  & $8.1 \times 10^1$              & $6.1 \times 10^2$           & $3.8 \times 10^1$\\
\bCubeFour   & $1.6 \times 10^2$              & $ 7.9\times 10^2$           & $5.3 \times 10^1$ \\ 
\bCubeEight    & $3.3 \times 10^3$            & ---                         & $2.2 \times 10^3$\\
\cCubeOne    & $2.6 \times 10^{-1}$              & $9.0 \times 10^0$           & $3.0 \times 10^0$ \\  
\dblpCubeOne & $5.0 \times 10^0$              & $1.4 \times 10^1$          &   $3.7 \times 10^0$\\
\dblpCubeThree& $7.3\times 10^0$              & $3.1\times 10^1$          &  $5.3 \times 10^0$     \\
\nfCubeOne   & $4.2 \times 10^2$              & $1.5 \times 10^3$           & $3.8 \times 10^2$\\
\nfCubeThree & $5.8 \times 10^1$              & $3.6 \times 10^2$           & $3.5 \times 10^1$\\
\wCubeFour  & $1.3 \times 10^3$              & $7.4 \times 10^3$          & $1.4 \times 10^3 $\\\hline        
\end{tabular}

\end{table}
The preprocessing of the cubes was timed separately from
\diamondcube\ building.  Preprocessed data could be used many times,
varying the value for $k$, without incurring additional  preparation
costs.  Table~\ref{tb:preproc-times} summarises the times needed to
preprocess each cube in preparation for the algorithms that were run
against it.  Using MonetDB 
 was in most cases, the most efficient method.  
  For  comparison, sorting the Netflix comma-separated data
file---using the GNU sort utility---took $5.3\times 10^2$~seconds. 
 
\subsection{Finding  \texorpdfstring{$\kappa$}{kappa} for \textsc{count}-based \Diamondcube s}
\label{sec:findingkappa}

Using  Proposition~\ref{prop:kappa}, the
$\kappa$-carat \diamondcube\ was built for each of the data sets.  The
initial guess~($k$) for $\kappa$ was the value calculated using
Proposition~\ref{prop:kappa}.  Then  $k$ was
repeatedly doubled until an empty cube was returned and  a tighter
range for $\kappa$  had been established. 
Next a simple binary
search, which used the newly discovered lower and upper bounds as the
end points of the search space, was executed.  Each time a non-empty \diamondcube\ was returned, it was used as the input to the next iteration of the search. When the guess overshot $\kappa$ and  an empty \diamondcube\ was returned, the most recent non-empty  cube was used as the input.

\begin{table}[htb]
\captionsetup[table]{position=top}
\caption{Iterations to convergence for \textsc{sum} and \textsc{count}-based \diamondcube s \label{tb:iters}}
\centering
\subfloat[The number of iterations and time (in seconds) it took to determine the $\kappa$-carat \diamondcube\  for
\textsc{count}-based \diamondcube s.\vspace{10pt}]{%
\label{tb:count-diamond-specs}%
\begin{tabular}{l|c|rr|rr|l}\hline
Algorithm &     cube           & \multicolumn{2}{c|}{iterations}           & \multicolumn{2}{c|} {value of $\kappa$}& time \\
                         && $\sum n_i$   & actual  & est.      & actual          & (in seconds) \\\hline
\multirow{4}{*}{\imd} 
          &\twCubeOne     &   91         & 6       &  19       &  38             & 1.0$\times10^{-2}$\\
          &\nfCubeThree   &   473\,753   & 17      & 39        & 272             & 6.0$\times10^0$\\
          & \dblpCubeOne  &  645\,739    & 23      & 3         & 30              & 3.0$\times 10^{-1}$\\
          & \dblpCubeThree&  689\,519    & 26      & 4         & 43              & 8.0$\times 10^{-1}$\\
          &\bCubeTwo      & 27\,042      & 16      & 884       & 7\,094          & 5.0$\times10^0$\\
          &\bCubeThree    & 29\,078      & 19      & 1\,098    & 8\,676          & 6.7$\times10^0$\\
          &\cCubeOne      & 5\,607       & 8       &  282      &  672            & 1.5$\times 10^{-1}$\\\hdashline[1pt/1pt] 
\multirow{7}{*}{\emd}
          &\nfCubeOne     & 484\,141     & 19      &  197      &  1\,004         & 3.4$\times10^1$\\
          &\wCubeFour     & 48\,654      & 26      & 2\,550    & 4\,554          & 6.2$\times10^2$\\
          &\bCubeOne      & 31\,634      & 12      & 1\,723    &  14\,383        & 1.6$\times10^1$\\
      
          &\bCubeFour     & 30\,417      & 12      & 1\,347    & 10\,513         & 1.2$\times10^1$\\
          &\bCubeEight    & 6\,335       &  5      & 57\,668   & 112\,232\,566   & 1.0$\times10^3$\\\hline
\end{tabular}
}\qquad
\centering
\subfloat[The number of iterations and time (in seconds) it  took to determine the  $\kappa$-carat \diamondcube\ on \textsc{sum}-based \diamondcube{}s.
  The estimate for $\kappa$ is the tight lower bound from
  Proposition~\ref{prop:kappa-sum-lower-bounds}.
]{\label{tb:sum-diamond-specs}%
\begin{tabular}{c|rr|rr|l}\hline
  cube          &\multicolumn{2}{c|}{iterations}    & \multicolumn{2}{c|} {value of $\kappa$}&time \\
                &$\sum n_i$         & actual        & est.          & actual                       & (in seconds)\\\hline
\bCubeFive      & 31\,634           &4              & 729           & 25\,632                       & 5.6$\times 10^1$\\

\cCubeTwo       & 504                & 5             &  1\,853      &  3\,600\,675                 & 6.0$\times10^{-1}$\\
\dblpCubeTwo   & 645\,739           & 7             &  113          & 119                          & 7.5$\times 10^{-1}$\\
\nfCubeTwo      &484\,141           & 40            &  5            &  3\,483                      & 1.6$\times 10^2$\\ 
\ssbCubeOne     & 22\,000           & 8             & 2\,124\,269   &1\,581\,756\,429               & 4.6$\times 10^0$\\
\twCubeThree    & 674                & 3             &  85          &  85                          & 4.3$\times 10^{-2}$\\ 
\wCubeTwo       & 48\,654           &19             & 32            & 20\,103                       & 1.9$\times 10^3$\\ 
 \hline
\end{tabular}
}
\end{table}

Statistics are provided in Table~\ref{tb:count-diamond-specs}. The
estimate of $\kappa$ comes from Proposition~\ref{prop:kappa} and the
number of iterations recorded is the number 
used by Algorithm~\ref{algo:shrinkingInput} to compute
the
$\kappa$-carat \diamondcube{} given $\kappa$. 
The estimates for $\kappa$ vary between 4\%
and 50\%  of the actual value and there is no clear pattern to
indicate why this  might be. Two very different cubes both have
estimates that are  50\% of the actual value: \twCubeOne{}, a  small cube of  less than 2\,000 cells  and low  dimensionality,  and \wCubeFour{}, a large cube of  $1.23 \times 10^8$ 
  cells with moderate dimensionality.  We experimented with sampling to provide an improved estimate  for $\kappa$.  We chose 10~independent samples for each of 1\%, 5\% and 10\% 
using 3 of our largest cubes, which
have different characteristics. We computed $\kappa$ for each of these cubes.    In Table~\ref{tb:sampling-for-kappa} we see that even with just 1\% of the data, the estimate for $\kappa$ is very close to 1\% of the actual value, and, therefore, provides a better estimate than that of the bound (used as an estimate) given by Proposition~\ref{prop:kappa}.  
Since it is an estimate, rather than a bound, we can test whether the diamond is empty
for this value.
Depending on the outcome, our estimate  can then be used as an upper or a lower bound for the binary search.

\begin{table}

\captionsetup[table]{position=top}
\caption{\label{tb:sampling-for-kappa}We can use the value of  $\kappa$ from uniform samples of the data to estimate $\kappa$ for the entire cube.}
\begin{center}
\begin{tabular}{lll|ll|ll|ll}\hline
Cube         & Actual $\kappa$ & Estimate & \multicolumn{6}{c}{Sample Size (10 samples)}\\
             &                 & from Prop~\ref{prop:kappa}& 1\%       & 1\%    & 5\% & 5\%     & 10\%  & 10\%\\
             &                 &                           & min       & max    & min      & max           &min     & max \\\hline
\nfCubeOne   & 1004            & 197                       & 11        & 11            & 51       & 51     &101     &102     \\
\wCubeFour   & 4554            & 2550                      & 44        & 45            & 223      & 224    &450     & 451    \\
\bCubeOne    & 14\,383           & 1\,723                      & 144       & 147           & 717      & 728           & 1\,428       & 1\,440 \\\hline
\end{tabular}
\end{center}
\end{table}

\subsection{Finding  \texorpdfstring{$\kappa$}{kappa(C)} for \textsc{sum}-based \Diamondcube s}
\label{sec:findingkappaforsumbased}

From  Proposition~\ref{prop:kappa-sum-upper-bounds}, we have that $\min_i(\max_j(\sigma(\mathrm{slice}_j(D_i))))$ is 
 an upper bound on $\kappa$ for  any \textsc{sum}-based \diamondcube\  and from
 Proposition~\ref{prop:kappa-sum-lower-bounds} a lower bound  is the
 maximum value stored in any cell.  Indeed, for cube
 \twCubeThree\  the  lower bound  is the $\kappa$ value.  For this reason, the approach to finding $\kappa$ for the \textsc{sum}-based \diamondcube{}s  varies slightly   in that the first guess for $k$ should be the lower bound + 1.  
 If this returns a non-empty \diamondcube, then a binary search over
 the range 
from the lower bound + 1 to the upper bound
is used to find $\kappa$.
 Statistics are given in Table~\ref{tb:sum-diamond-specs}. 

\subsection{Comparison of Algorithm Speeds\label{sec:algo-times}}

\begin{figure}
  \begin{center}
      
      \includegraphics[width =.49\textwidth]{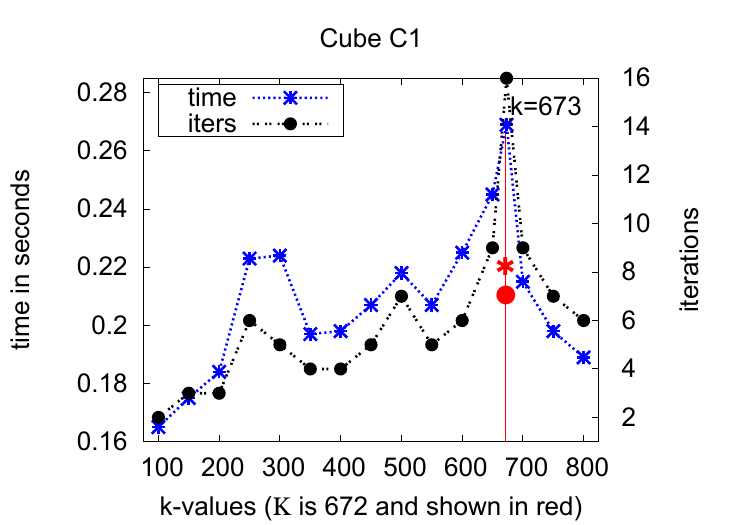}
\includegraphics[width=.49\textwidth]{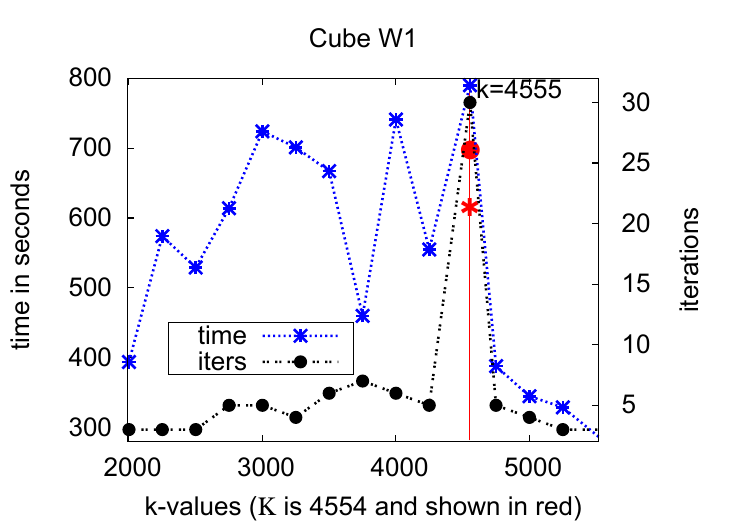}
 
\caption[Times and iterations needed to generate \diamondcube s with
  different $k$-values.]{%
Times and iterations needed to generate
  \diamondcube s with different $k$-values. In each case
  more   iterations are required for $k$-values that slightly
  exceed $\kappa$. 
 The increase from
$k=672$ to $k=673$ is particularly evident.
\label{fig:iters-vs-time}}
  \end{center}
\end{figure}

In Table~\ref{tb:iters} we report times for processing  the  $\kappa$-carat \diamondcube\ for each
of nineteen cubes.  Our implementation processes cubes of 20\,000\,000 -- 
40\,000\,000  records in less than a
minute. 

Table~\ref{tb:sql-slowdown} compares the speeds of
Algorithms~\ref{algo:shrinkingInput}  or 
\imd\  with
Algorithm~\ref{algo:bothSQL-and-loop}.  
  Times were averaged over five runs  and then normalised against \emd\
  or \imd \@. We see that \emd\ and \imd\ 
effect greater speed-up as the cube size increases and the cube density decreases. For
example,  \imd\  is 4 times faster on the small, dense cube,
\twCubeOne{}, and~\ref{algo:shrinkingInput} is  500~times faster on the more
sparse cube, \nfCubeThree. 

 Although the \diamondcube\ dice operation is inherently a row-wise computation, we find MonetDB can be much faster than MySQL (up to 23~times faster).  MonetDB was able to complete even the most difficult computation in no more than 3.6~hours whereas MySQL needed more than 19~hours in this case. However, our Java code did it in 17~minutes. The amount of available memory affects the running times of all our algorithms.   We restricted the amount of  memory to 2\,GiB
 and processed  cube \nfCubeOne :
\begin{itemize}
\item \emd\ was 2.6~times slower (1.5~minutes).
\item  MonetDB was 10~times slower (9~hours).   
\item MySQL was forcibly terminated after 23~hours.
\end{itemize}   

Neither the initial file size nor the number of cells pruned from each
$k$-carat \diamondcube\ alone explains the time necessary to generate each
\diamondcube. In an earlier implementation of the \diamondcube\ dicing algorithm, we had
observed that the time expended was proportional to the number of
cells processed.  This is not as evident in the current implementation, where a new
file is  written  when 50\% of the cells are marked for deletion,
instead of at every iteration.  However,  for all cubes more iterations and time are
required to process $k$-values that only slightly exceed $\kappa$.  
In one instance (see Fig.~\ref{fig:iters-vs-time}), we need more than twice the number of iterations and nearly twice
the time to compute the (empty) $k=673$-carat \diamondcube\  than to compute
the $\kappa=672$-carat \diamondcube.  In both examples presented in Fig.~\ref{fig:iters-vs-time} (cubes \cCubeOne{} 
and \wCubeFour{}), the number of iterations needed to compute the $k$-carat  \diamondcube\ for a value of $k$ either 20\% above or below $\kappa$ is at least half of the  number of iterations observed for $k=\kappa + 1$.  Similarly, 30\% more time is required to
process $\kappa +1$  for cube \wCubeFour.
Intuitively, one should not be
surprised that more iterations, and thus time, 
are required when $k \approx \kappa$: attribute values that are almost in
the \diamondcube\ are especially sensitive to  other attribute values
that are also almost in the \diamondcube.

\begin{table}[htb]
\centering
\caption[
Slowdown of the SQL algorithms]{%
Relative slowdown of the SQL algorithm compared to \emd{} or \imd{}. Times were averaged over ten runs.  
MySQL processing for cube \bCubeEight\ was forcibly terminated after 19 hours ($\bigotimes$).\label{tb:sql-slowdown}}
\begin{tabular}{lllll|ll}\hline
             &                          &           & \multicolumn{2}{c|}{SQL(s)}&     \multicolumn{2}{c}{Ratio}       \\
  Cube       &\imd\  (s)                & \emd\ (s)         & MySQL              & MonetDB          & MySQL              & MonetDB\\\hline
\cCubeOne    & 4.0$\times10^{-1}$        &  ---              & 1.1$\times10^{1}$   & 1.0$\times 10^1$  & 2.7$\times10^1$     & 2.5 $\times 10^1$\\
\dblpCubeOne & 3.0$\times 10^{-1}$       &  ---              & 7.4$\times 10^{1}$  & 7.9$\times 10^1$  & 2.5$\times 10^2$    & 2.6$\times 10^2 $\\
\dblpCubeThree & 8.0$\times 10^{-1}$     &  ---              & 1.2$\times 10^2$    & 9.3$\times 10^1$ & 1.5$\times 10^2$     & 1.2$\times 10^2$ \\
\bCubeTwo    &  ---                     & 8.0$\times10^{0}$  & 1.9$\times10^{3}$   & 8.2$\times 10^1$  & 2.4$\times10^1$     & 1.0 $\times 10^1$\\
\bCubeThree  &  ---                     & 1.2$\times10^{1}$  & 2.7$\times10^{3}$   & 1.1$\times 10^2$  & 2.3$\times10^2$     & 9.0 $\times 10^0$\\
\bCubeFour   &  ---                     & 1.2$\times10^{1}$  & 3.5$\times10^{3}$   & 2.2$\times 10^2$  & 2.9$\times10^2$     & 1.8 $\times 10^1$\\
\bCubeEight  &  ---                     & 1.0$\times 10^{3}$ &  $\bigotimes$      & 1.3$\times 10^{4}$ & ---                & 1.3 $\times 10^1$ \\
\nfCubeThree & 6.0$\times10^{0}$         & ---               & 1.3$\times10^{3}$   & 3.3 $\times 10^2$ & 2.0$\times10^2$     & 5.5 $\times 10^1$\\
\nfCubeOne   &  ---                     & 3.4$\times10^{1}$  & 1.9$\times 10^{4}$  & 5.0$\times 10^3$  & 5.6$\times 10^2$    & 1.5 $\times 10^2$ \\\hline
\end{tabular}

\end{table}

\subsection{\label{subsec:diamond-size}\Diamondcube\ Size and Dimensionality}
The size (in cells)
 of the $\kappa$-carat \diamondcube\ of the high-dimensional cubes is 
large, e.g.\ the $\kappa$-carat \diamondcube\
for \bCubeEight\ captures  30\% of the data. How can we explain this? Is this  property a function of the number of dimensions?   To answer this question 
 the $\kappa$-carat \textsc{count}-based \diamondcube\ was generated for each of the synthetic  cubes (except \ssbCubeOne ). 
  Estimated $\kappa$, its real value  and the size in cells for each cube are given in Table~\ref{tb:synth-kappa}. 
The $\kappa$-carat \diamondcube\ captures 98\% of the data in cubes
\uCubeOne{}, \uCubeTwo\ and \uCubeThree ---dimensionality has no effect
on \diamondcube\ size for these uniformly distributed data sets.
Likewise,  dimensionality did not affect the size of the
$\kappa$-carat \diamondcube\ for the skewed data cubes as it captured
between 23\% and 26\% of the data in cubes \sCubeOne{}, \sCubeTwo\ and
\sCubeThree\ and between 12\% and 17\% in the other cubes.  These
results indicate that the dimensionality of the cube does not affect 
how much of the data is captured by the \diamondcube\ dice.  

\begin{table}[htb]
\centering
\caption{\label{tb:synth-kappa}High dimensionality does not affect \diamondcube\ size.}
\begin{tabular}{crrrrrc}\hline
Cube        & dimensions & iters     & \multicolumn{2}{c}{value of $\kappa$}    & size (cells)  & \% captured \\%
            &            &           & est. & actual   &                &      \\\hline
\uCubeOne   &   3        &  6        & 89        & 236      & 982\,618       &   98 \\
\uCubeTwo   &   4        &  6        & 66        & 234      & 975\,163       &   98 \\
\uCubeThree &  10        &  7        & 25        & 229      & 977\,173       &   98 \\\hdashline[1pt/1pt]
\sCubeOne   &  3         &  9        & 90        & 1141     & 227\,527       &   24 \\
\sCubeTwo   &  4         &  14       & 67        & 803      & 231\,737       &   23 \\
\sCubeThree &  10        &  14       & 25        & 208      & 260\,864       &   26 \\\hdashline[1pt/1pt]
\ssCubeOne  &  3         & 18       & 11        & 319      & 122\,878       &   12 \\
\ssCubeTwo  &  4         & 19       & 7         & 175      & 127\,960       &   13 \\
\ssCubeThree&  10        &  17       & 1         & 28       & 165\,586       &   17\\ \hline
   \end{tabular}
\end{table}

\subsection{Iterations to Convergence}

\begin{figure}
\centering
\includegraphics[width=0.75\textwidth]{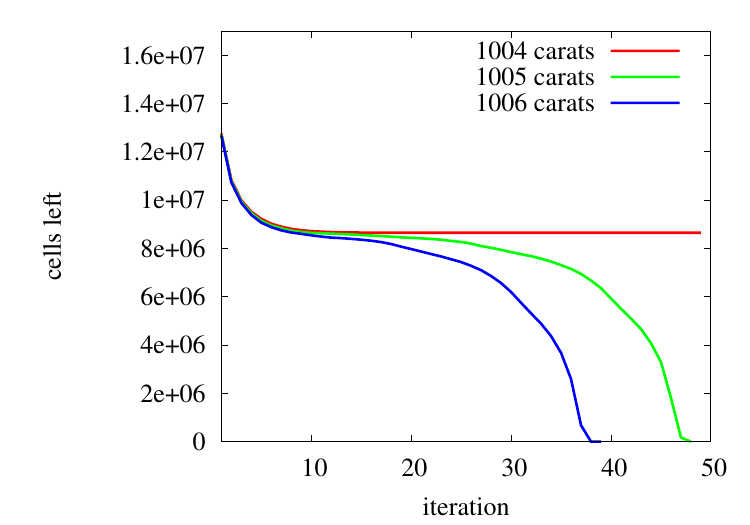}
\caption{\label{fig:cells-left-netflix}
 Cells remaining after each iteration of  Algorithm~\ref{algo:shrinkingInput}  for $k$ = 1004, 1005 and 1006 on
 cube \nfCubeOne{}.}\end{figure}

In Section~\ref{subsec:complexity} we observed that in the worst case it
could take $\Theta(\sum_i n_i)$ iterations before the
\diamondcube\ cube stabilised.  In practise this was not the
case. (See Table~\ref{tb:iters} and
Fig.~\ref{fig:iters-vs-time}).
 All cubes converged to the $\kappa$-carat \diamondcube\ in less than 1\% of
 $\sum_i n_i$,  with the exception of the small cube \twCubeOne{},  which took less than 7\% $\sum_i n_i$.   Algorithm~\ref{algo:shrinkingInput} required  19~iterations and  34~seconds\footnote{Times were averaged over 10~runs.} to compute the 1\,004-carat $\kappa$-carat \diamondcube\ for \nfCubeOne\ 
 and it took 50 iterations and an average of 72~seconds 
 to determine that there 
is  
no 1\,005-carat \diamondcube{}. 

For several values of $k$, we measured the number of cells remaining in cube \nfCubeOne\  after each iteration of Algorithm~\ref{algo:shrinkingInput}, in order to see
how quickly the \diamondcube\ converges to an empty \diamondcube\ when $k$  exceeds $\kappa$. Fig.~\ref{fig:cells-left-netflix} shows the number of cells present in the \diamondcube\ after each iteration for 1\,004--1\,006~carats. 
The curve for 1\,006 reaches zero first, followed by that for 1\,005. Since
$\kappa = 1\,004$, that curve stabilises at a nonzero value. 
It takes longer to reach a critical point when $k$  only slightly
exceeds $\kappa$. 

 The number of iterations
required until convergence for all our  synthetic cubes was also
far smaller than the upper bound,
e.g.\ cube \sCubeThree :
35\,616~(upper bound) and 14~(actual).  We had expected to see the
uniformly distributed data taking longer to converge than the skewed
data. This was not the case: in fact the opposite behaviour was observed.  (See Table~\ref{tb:synth-kappa}.)  For cubes \uCubeOne{}, \uCubeTwo\ and \uCubeThree\ the \diamondcube\ captured 98\% of the cube: less than 23\,000 cells were removed, suggesting that they started with a structure very like a \diamondcube\ but for the  skewed data cubes---\sCubeOne{}, \sCubeTwo{}, \sCubeThree{}, \ssCubeOne{}, \ssCubeTwo\ and \ssCubeThree ---the \diamondcube\ was more ``hidden''. 

\section{Related Work\label{sec:related-work}}

There are other multidimensional operations that can be useful to an analyst, such as Skyline (Section~\ref{sec:skyline}),  Nearest Neighbours and Outliers (Section~\ref{sec:subsampdatabase}). However, they differ from \diamondcube s in several ways. Except for \textsc{iterative pruning}, none
is a form of dicing, that is they do not select interesting attribute values, and some assume that attribute values are ordered or that we have a distance measure between records. In the rest of this section, we review these related queries in more detail.

\subsection{Trawling the Web for Cyber-communities \label{subsec:trawling}}

A specialisation of the \diamondcube\ cube is found in Kumar~et~al.'s work searching for emerging social networks on the Web~\cite{313082}.  Our approach is a generalisation of  their two-dimensional \textsc{iterative pruning} algorithm. \Diamondcube s are inherently multidimensional.
 Kumar et al.~\cite{313082} model the Web as a directed
graph and seek large dense bipartite sub-graphs. 
 A bipartite graph is dense if most of the vertices in the two
 disjoint sets,  $U$ and $V$, are connected. Kumar et al.~hypothesise that the signature of an emerging Web community contains at least one ``core'', 
 which is a \emph{complete} bipartite sub-graph with at least $i$ vertices from $U$ and $j$ vertices from $V$. In their model, the vertices in $U$ and $V$ are Web pages and the edges are links from $U$ to $V$. Seeking an ($i,j$) core is equivalent to seeking a \emph{perfect} two-dimensional  \diamondcube\ cube (all cells are allocated).  
Their \textsc{iterative pruning} algorithm  is a specialisation of the basic algorithm we use to seek  \diamondcube s: it is restricted to two dimensions and is used as a preprocessing step to prune data that cannot be included in the ($i,j$) cores. 
A multidimensional extension of their algorithm proved to consume too much
memory and run too slowly. See~\cite{hazel-thesis}.

\subsection{Skyline Operator}
\label{sec:skyline}
   The Skyline
   operator~\cite{DBLP:conf/icde/BorzsonyiKS01,springerlink:10.1007/978-1-4419-6045-0_16} 
seeks a set of points  where each point is not ``dominated'' by some others:  a point is
 included in the skyline if it is as good or better in all dimensions
 and better in at least one dimension.   
    Attributes, e.g.\ distance
 or cost,  must be ordered.

Skyline queries  have been adapted to the  OLAP context~\cite{5230659} as the Multi-Objective OLAP (MOOLAP) framework. 
Like diamonds, the goal is to allow an analyst to focus on
\emph{interesting} data. For example, the analyst might be interested
in stores that have either high profitability or high volume of sales (ideally both).
Like diamonds, MOOLAP assumes that user-provided aggregators are 
monotone (e.g., like \textsc{sum}). In contrast to diamonds, MOOLAP results do not form a dice.

\subsection{Sub-sampling with Database Queries}
\label{sec:subsampdatabase}
Relational Database Management Systems~(RDBMS) have optimisation routines that are especially tuned to address both basic and more complex SELECT \ldots FROM \ldots WHERE \ldots\ queries.  
However, there are some classes of queries  that are 
 difficult to express in SQL, or that execute slowly,  because suitable algorithms  are not available
 to the underlying query engine.  
 Besides skyline, they 
include 
top-$k$ and nearest-neighbour queries. 

\paragraph{\texorpdfstring{Top-$k$}{Top-\emph{k}}}

Another query, closely related to the skyline query, is that of finding the ``top-$k$''
 data points.  For example, we may seek the ten most popular 
 products sold in a store. While this can help the work of the
 analyst, browsing only the top-$k$ results can also 
 improve performance~\cite{donjerkovic1999pot} by reducing the size of the result set. 

\paragraph{Nearest Neighbours}

One of the most common multidimensional queries is the nearest neighbour query, which
seeks  elements that are ``close'' to a provided target. For example, given a set of users, we might
seek users who  have a profile similar to the current user.
A common query asks to find the $k$~nearest neighbours (kNN),
that is, $k$ neighbours that are as close as possible
to the target. 

Reverse nearest neighbours~\cite{Korn:2000:ISB:342009.335415}
starts with a given element and asks which possible targets
would have this element in the nearest neighbours. For example,
imagine that customers only visit one of the 10~nearest stores.
Given a customer, which store locations would attract him?
Nearest neighbour queries require a specific distance measure.

\paragraph{Outlier Identification}

Another frequent type of query in multidimensional data analysis 
 is outlier identification. For example, we might seek
elements that are far from most other data
points~\cite{Knorr:1998:AMD:645924.671334}.
Sarawagi et al.~\cite{Sarawagi:1998:DEO:645338.650401} define outliers
in the OLAP context as deviations from anticipated values (computed from a model).   Their approach
 requires learning a model from the data so that anticipated values can be computed. It also serves to highlight possibly interesting data in a large data cube. 

\paragraph{Iceberg Queries}

The iceberg query introduced by Fang et al.~\cite{iceberg98}
eliminates aggregate values below some specified threshold.
For example, if we have sales data by month and by store, we might require sales to exceed a threshold: only pairs (month, store) above the threshold are kept. These might be considered  interesting by the analyst. 
In contrast, \diamondcube\ dicing applies several thresholds  simultaneously. In effect, we could consider
\diamondcube\ dicing as the simultaneous application of several interacting 
 iceberg  thresholds. 

\subsection{Formal Concept Analysis\label{subsec:fca}}

In Formal Concept Analysis~\cite{102547,DBLP:journals/ci/GodinMA95} a
Galois (concept) lattice is built from a binary relation.  It is used
in machine learning to identify conceptual structures among data sets. For example, a concept can be formed from a set of documents and the set of search terms those documents match.  We put a value of 1 in a cell if the corresponding document contains the corresponding term, otherwise we leave the cell unallocated. A Galois concept in this case would be a list of documents and a list of terms such that every document contains every term in the list, and every term is contained in every document. Just like a \diamondcube, Galois concepts must be maximal: there cannot be another Galois concept that contains all the documents and terms, and some more.
 Given the data in Fig.~\ref{tb:2d}, the smallest concept including document 1 is the one with documents \{1, 2\} and search terms \{A,B,C,D\}.  Concepts can be partially ordered by inclusion and thus can be represented by a lattice as in Fig.~\ref{fig:galois}

 Galois lattices are  related to diamond cubes: in effect, a Galois concept is  a \emph{perfect} \textsc{count}-based \diamondcube --- one with all cells allocated --- in a two-dimensional setting. 
 Though 
  formal Concept Analysis
 is  typically restricted to two dimensions, 
Cerf et al.~\cite{1497580}
 generalise formal concepts by presenting an algorithm that is applied to more than two dimensions. Their definition of a closed $n-$set---a formal concept in more than two dimensions---states that each element is related to all others in the set and no other element can be added to this set without breaking the first condition.  It is the equivalent of finding a perfect \diamondcube\ in $n$ dimensions. 

In real data sets, we are unlikely to find large perfect \diamondcube{}s though we can find many small ones, especially if there are many dimensions. Galois concepts are brittle: a single omitted cell is sufficient to make a concept disappear. Thus, for an analyst, Galois concepts may be difficult to use. 

\begin{table}[htb]\centering
\caption{A 3-dimensional relation with closed 3-set \{($\alpha,\gamma)(1,2)(A,B)$\}.\label{tb:n-closed-set} } 
\begin{tabular}{|cc|ccc|ccc|ccc|}\hline
                                                                           &  & \multicolumn{3}{c|}{dimension 3} &\multicolumn{3}{c|}{dimension 3}&\multicolumn{3}{c|}{dimension 3}\\
                                                                           &  &  A & B & C & A & B & C & A & B &C\\\hline
\multirow{4}{3mm}{\begin{sideways}\parbox{18mm}{dimension 2}\end{sideways}}&1 &   1 & 1 & 1 & 1 & 1 & 1 & 1 & 1 &   \\
                                                                           &2 &   1 & 1 &   & 1 & 1  &   & 1 & 1 &  \\
                                                                           &3 &     & 1 &   &   &   & 1 & 1 &   & 1 \\
                                                                           &4 &     &   & 1 & 1 &   & 1 & 1 & 1 & 1\\\hline
                                                                           &  &\multicolumn{3}{c|}{$\alpha$}&\multicolumn{3}{c|}{$\beta$} & \multicolumn{3}{c|}{$\gamma$}\\\cline{3-11}
                                                                           &  &\multicolumn{9}{c|}{dimension 1}\\\hline
\end{tabular}
\end{table}  

\begin{figure}
\centering

\subfloat[A $ 3 \times 3$ \diamondcube\ is embedded in this binary relation.
]{\label{tb:2d}
\begin{tabular}[b]{|cc!{\color{black}\vline\hspace{1pt}}ccccc|}\hline
                                                                           &   &\multicolumn{5}{c|}{Search Terms}\\
                                                                           &   & A & B & C & D & E\\\hline\\[-1.25em]
 \multirow{6}{4.75mm}{\begin{sideways}\parbox{15mm}{Documents}\end{sideways}} & 1 &\cellcolor{grayish}{1} &\cellcolor{grayish}{1} &\cellcolor{grayish}{1} & 1 & \\
                                                                           & 2 &\cellcolor{grayish}{1} &\cellcolor{grayish}{1} &\cellcolor{grayish}{1} & 1 & \\
                                                                           & 3 &\cellcolor{grayish}{1} & \cellcolor{grayish}{1} &\cellcolor{grayish}{1} &  & 1\\
                                                                           & 4 & 1 &  & 1 &  & \\
                                                                           & 5 &  & 1 & 1 &  & 1\\\hline
\end{tabular}
}
\qquad
\subfloat[Galois lattice.  Each element (concept) in the lattice is defined by its extent and intent.]{%
\includegraphics[height=.25\textheight]{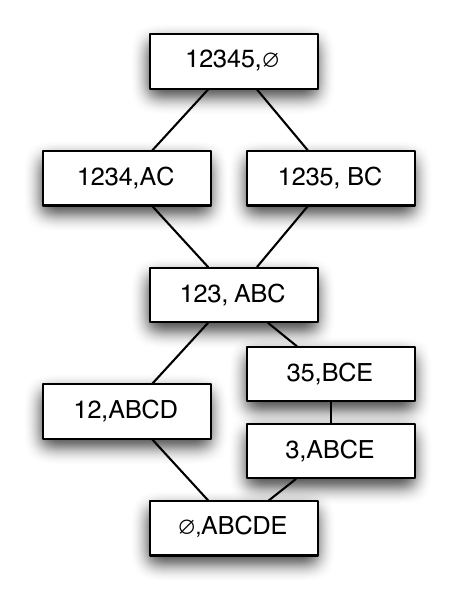}
\label{fig:galois}
}
\caption{Documents and search terms in an information retrieval system and the corresponding Galois lattice. \label{fig:fca}}
\end{figure}

\section{Conclusion\label{sec:conclusion}}
  
We presented a formal analysis  of the \diamondcube\ cube. We have shown that, for the  parameter $k$ associated with each dimension in every data cube, there is only one $k_1, k_2, \ldots k_d$-carat \diamondcube.  By varying the $k_i$'s we get a collection of \diamondcube s for a cube.  We established upper and lower bounds on the parameter $\kappa$ for both  \textsc{count}  and \textsc{sum}-based \diamondcube\ cubes. 

 We have designed, implemented and tested algorithms  to compute \diamondcube{}s
 on real and synthetic data sets.  Experimentally, the algorithms bear
 out our theoretical results. An unexpected experimental result is that the number of iterations
  required to process the  \diamondcube{}s with $k$  slightly greater
  than $\kappa$ is often twice  that required to process the
  $\kappa$-carat \diamondcube{}. 
 This also results in an increase in
  running time.

 We have shown that computing
 \diamondcube s for large data sets is feasible.  
\ref{algo:shrinkingInput} fared better on large, sparse data cubes 
 than other approaches and our results confirm that this algorithm
 is scalable.

\paragraph{Future Research Directions} 

 Although it is faster to  compute a \diamondcube\ cube using our
 implementation than using  the standard relational DBMS  operations,
the speed does not conform to the OLAP goal of
 near constant time query execution.   Different approaches could be
 taken to improve execution speed: compress the data so that more of
 the cube can be retained in memory; use multiple processors in
 parallel; or, if an approximate solution is sufficient,  we might
 process only a sample of the data.  These are some of the ideas to be
 explored in future work.

Data cubes are often organised with hierarchies of relationships within
dimensions.  For example, a \emph{time} dimension may include aggregations
for year, month and day.  Our current  work does not address the issue
of hierarchies and how they might be exploited in the  computation of \diamondcube s.  This is also a potential avenue for future work.

\bibliographystyle{model1-num-names}

\appendix 
 
\renewcommand{\thetheorem}{\Alph{section}.\arabic{theorem}}
\renewcommand{\thecorollary}{\Alph{section}.\arabic{corollary}}
\renewcommand{\theproposition}{\Alph{section}.\arabic{proposition}}
\renewcommand{\thelemma}{\Alph{section}.\arabic{lemma}}
\section{Bounding the carats for \textsc{sum}-based \Diamondcube s \label{sec:sum-diamond-properties}}

For \textsc{sum}-based \diamondcube s, the goal is  to capture a large fraction of
the sum. The statistic, $\kappa$, of a
\textsc{sum}-based \diamondcube\ is the largest sum for which there
exists a non-empty \diamondcube : every slice in every dimension has
sum at least $\kappa$ (see Section~\ref{sec:APrioriBoundsontheCarats}).  Propositions~\ref{prop:kappa-sum-lower-bounds} and~\ref{prop:kappa-sum-upper-bounds}  give tight  lower and upper bounds respectively for $\kappa$.

\begin{proposition}
\label{prop:kappa-sum-lower-bounds}Given a non-empty cube $C$ and the aggregator \textsc{sum},
 a tight lower bound on $\kappa$ is the value of the maximum cell ($m$). 
\begin{proof}  The $\kappa$-carat \diamondcube{}, by definition, is non-empty, so it follows that when the $\kappa$-carat \diamondcube\ comprises a single cell, then $\kappa$ takes the value of the maximum cell in $C$.
When the $\kappa$-carat \diamondcube\ contains more than a single cell, $m$ is still a lower bound: either $\kappa$ is greater than or equal to $m$.\end{proof}
\end{proposition} 

Given only the size of a \textsc{sum}-based diamond cube (in cells), there is no upper bound
on its number of carats. However, given its sum, say $S$, then it cannot
have more than $S$~carats.  We can determine a tight upper bound on $\kappa$ as the following proposition shows.

\begin{proposition}
\label{prop:kappa-sum-upper-bounds}
A tight upper bound for $\kappa$ is 
\begin{equation*}
 \min_{i}(\max_j(\textsc{sum}(\mathrm{slice}_j(D_i))))
\mathrm{\ for\ }i \in \{1,2,\ldots,d\} \mathrm{\ and\ } j \in\{1,2,\ldots, n_i\}.
\end{equation*}

\begin{proof}
Let $X$ = $\{ \textrm{slice}_j(D_i)\,|\,
\textsc{sum}(\textrm{slice}_j(D_i)) =
\max_k(\textsc{sum}(\textrm{slice}_k(D_i))\}$  then there is one slice
$x$ whose \textsc{sum}($x$)  is smaller than or equal to all other
slices in $X$.  Suppose $\kappa$ is greater than \textsc{sum}($x$)  then it follows that all slices in this $\kappa$-carat \diamondcube\ must have \textsc{sum} greater than  \textsc{sum}($x$).  However, $x$ is taken from $X$, where each member is the slice for which its \textsc{sum} is maximum in its respective dimension, thereby creating a contradiction.  Such a \diamondcube\ cannot exist.  Therefore, $ \min_{i}(\max_j(\textsc{sum}(\mathrm{slice}_j(D_i))))$ is an upper bound for $\kappa$.  

To show that  $ \min_{i}(\max_j(\textsc{sum}(\mathrm{slice}_j(D_i))))$ is also a tight upper bound we only need to consider a perfect cube where all measures are identical.\end{proof}

\end{proposition} 

\end{document}